\documentclass[11pt]{article}
\usepackage{fullpage}
\RequirePackage{amsthm}
\RequirePackage{amsmath}
\RequirePackage{amsfonts}		
\RequirePackage{graphics}
\RequirePackage{graphicx}
\RequirePackage{subfigure}
\RequirePackage{fancybox}		
\RequirePackage[letterpaper,hmargin=1.15in,vmargin=1.15in]{geometry}
\usepackage{times}

\newtheorem{theorem}{Theorem}[section]
\newtheorem{lemma}{Lemma}[section]
\newtheorem{claim}{Claim}[section]

\newtheorem{proposition}{Proposition}[section]

\newtheorem{observation}{Observation}[section]

\newcommand\card[1]{\left| #1 \right|} 
\newcommand\set[1]{\left\{#1\right\}} 
\newcommand\define{\stackrel{\Delta}{=}}

\newcommand\func[1]{ \it  #1 }

\begin{document}


\title{On the Fault Tolerance and Hamiltonicity of the Optical Transpose Interconnection System of Non-Hamiltonian Base Graphs}
      
     \author{Esha Ghosh \\
     Indian Institute of Technology, \\
     Madras, INDIA\\
     eshaghosh888@gmail.com \\
     \and
     Subhas K. Ghosh\\
     Siemens Information Systems Limited \\
     Bangalore, INDIA\\
     subhas.k.ghosh@gmail.com\\
     \and
     C. Pandu Rangan \\
     Indian Institute of Technology,\\
     Madras, INDIA\\
     prangan55@gmail.com
     }
    
\date{}

\thispagestyle{empty} 
\maketitle              


\begin{abstract}
Hamiltonicity is an important property in parallel and distributed computation. 
Existence of Hamiltonian cycle allows efficient emulation of distributed algorithms on a network wherever such algorithm exists for linear-array and ring, and can ensure 
deadlock freedom in some routing algorithms in hierarchical interconnection networks. 
Hamiltonicity can also be used for construction of independent spanning tree and leads to designing fault tolerant protocols. 
Optical Transpose Interconnection Systems or OTIS (also referred to as two-level swapped network)
is a widely studied interconnection network topology which is popular due to high degree of scalability, regularity, modularity and package ability.
Surprisingly, to our knowledge, only one strong result is known regarding Hamiltonicity of OTIS - showing that OTIS graph built of Hamiltonian base graphs are Hamiltonian. 
In this work we consider Hamiltonicity of OTIS networks, built on Non-Hamiltonian base and answer some important questions. 
First, we prove that Hamiltonicity of base graph is not a necessary condition for the OTIS to be Hamiltonian. We present an infinite family of Hamiltonian OTIS graphs composed on Non-Hamiltonian base graphs.
We further show that, it is not sufficient for the base graph to have Hamiltonian path for the OTIS constructed on it to be Hamiltonian.
We give constructive proof of Hamiltonicity for a large family of Butterfly-OTIS. This proof leads to an alternate efficient algorithm for independent spanning trees construction on this class of OTIS graphs. Our algorithm is linear in the number of vertices as opposed to the generalized algorithm, which is linear in the number of edges of the graph.
\end{abstract}

\section{Introduction}
\label{SEC1}
Optical Transpose Interconnection Systems (OTIS) is a widely studied interconnection network topology in parallel and distributed computing. OTIS(Swapped) Network was first
proposed by Marsden et al. in 1993 \cite{Marsden}. A number of computer architectures have subsequently been proposed in which the OTIS concept was used to connect new optoelectronic computer architectures
efficiently exploiting both optical and electronic technologies. In this architecture, processors are divided into groups (called clusters), where processors within the same 
group are connected using electronic interconnects, while optical interconnects are used for intercluster communication. The OTIS architecture has been used to propose 
interconnection networks for multiprocessor systems. Krishnamoorthy et al. have shown that the power consumption is minimized and the bandwidth rate is maximized
when the number of processors in a cluster equals the number of clusters \cite{Krishnamoorthy}. This is the key reason why we considered OTIS graphs consisting of $n$ clusters where each cluster is isomorphic to the base graph consisting of $n$ processors.

Fault tolerance is an important aspect of parallel and distributed systems. Two major kinds of hardware faults that can occur in networks are dead processor fault (due to failure
of processor or support chip) and dead interprocessor communication (due to failure of communication hardware). These faults can be abstracted as the failure of nodes and of
edges in the underlying network graph. Hence, an important parameter to measure the fault tolerance of a distributed system, is to count the number of Independent Spanning
Trees in the graph, which ensures the presence of parallel, node-disjoint paths between nodes of the network. Hamiltonicity is an important property in any hierarchical interconnection network that is closely
related to fault tolerance, as, the presence of a number of edge disjoint Hamiltonian cycles in a network implies twice that number of Independent Spanning Trees in
that network. Hamiltonicity is also important to ensure deadlock freedom in some routing algorithms \cite{Carpenter} and to allow efficient emulation of linear-array and ring algorithms.
Algorithms, such as all-to-all broadcasting or total exchange, relies on a Hamiltonian cycle for its efficient execution \cite{Parhami2}.

\subsection{Related Results}
OTIS (Swapped) have been extensively studied. Chen et al. have shown if the base graph is $k$ connected than OTIS will have $k$-vertex disjoint paths between any pair of vertices, and
this is defined as a notion of maximal fault tolerance by them \cite{Chen}.
Surprisingly, to our knowledge, very few results are known regarding 
Hamiltonicity of OTIS networks. The only significant result known about Hamiltonicity of OTIS, is by Parhami, that proves that OTIS networks built of Hamiltonian basis networks are Hamiltonian \cite{Parhami}.
The result by Hoseinyfarahabady et al. \cite{Hosein} shows that the OTIS-Network is Pancyclic and hence Hamiltonian,  if its base network is Hamiltonian-connected. However, by the fact that any Hamiltonian connected base graph is definitely Hamiltonian, this is a weaker result.

%

\subsection{Our Contribution}

We address some important aspect of Hamiltonicity on OTIS graphs. 
\begin{itemize}
\item{We investigate whether Hamiltonicity of base graph is also a necessary condition for the OTIS to be Hamiltonian.
We answer this in negative. }
\item{We further investigate whether it is sufficient for the base graph to have Hamiltonian path, for the OTIS to be Hamiltonian. We answer this in negative as well.}
\item{Two kinds of butterfly graphs known in literature. The first one is a 5-vertex graph (Fig 1) which is also known as bowtie graph. We consider the generalization of this butterfly/bowtie graph, 
where we consider two cycles $C_m, C_n$ connected at a cutvertex and denote in as $BF(n,m)$ . We consider $BF(n,m)$ as base network, and investigate the Hamiltonicity on the OTIS network.  To avoid ambiguity, we denote the OTIS network formed on  $BF(n,m)$ as Bowtie-OTIS.}
\item{A different type of butterfly graph of dimension $n$ is defined as a 4-regular graph, $BF(n)$ , on $n2^n$ vertices as follows \cite{Barth}:
\begin{itemize}
\item{The vertex set, $V(BF(n))$ is the set of couples $ (\alpha; x_{n-1}, \ldots, x_0)$, where $ \alpha \in \lbrace 0, \ldots, (n-1) \rbrace$ and $ x_i \in \lbrace 0,1 \rbrace, \forall i \in \lbrace 0, \ldots, (n-1) \rbrace$.}
\item{$ [(\alpha; x_{n-1}, \ldots, x_0), ({\alpha}^{\prime}; {x_{n-1}}^{\prime}, \ldots, {x_0}^{\prime})]$ is an edge of $BF(n)$ if $ \alpha^{\prime} \equiv \alpha+1 ($mod $n)$ and if $x_i = {x_i}^{\prime}  \forall i \neq \alpha^{\prime}$. }

\end{itemize}
We also investigate Hamiltonicity on the OTIS network built on this base network.}
\item{We give constructive proofs for Hamiltonicity, on Bowtie-OTIS of $BF(2m+1,2n+1)$ and $BF(2m+1,2k)$, where $ m, n, k \in \mathbb{N} $.
We also prove that number of edge-disjoint Hamiltonian Cycles possible on this class of Bowtie-OTIS is at most one. This construction leads to an efficient alternate linear time
Independent Spanning Tree construction algorithm on this class of Bowtie-OTIS graphs. This algorithm is linear in the number of vertices, as opposed to the generalized tree construction algorithm was proposed by Itai and Rodeh \cite{Itai}, 
which is linear in number of edges of the graph. So if we make $BF(2m+1,2n+1)$ or $BF(2m+1,2k)$ denser by introducing chords
inside the cycles $C_{2m+1}$, $C_{2n+1}$ or $C_{2k}$ such that at least one of the vertices retain degree 2, our algorithm shows better performance than the generalized one.}

 \end{itemize}

\begin{figure}[htpb]
   \begin{center}
     \resizebox{50mm}{!} {\includegraphics[scale=.2]{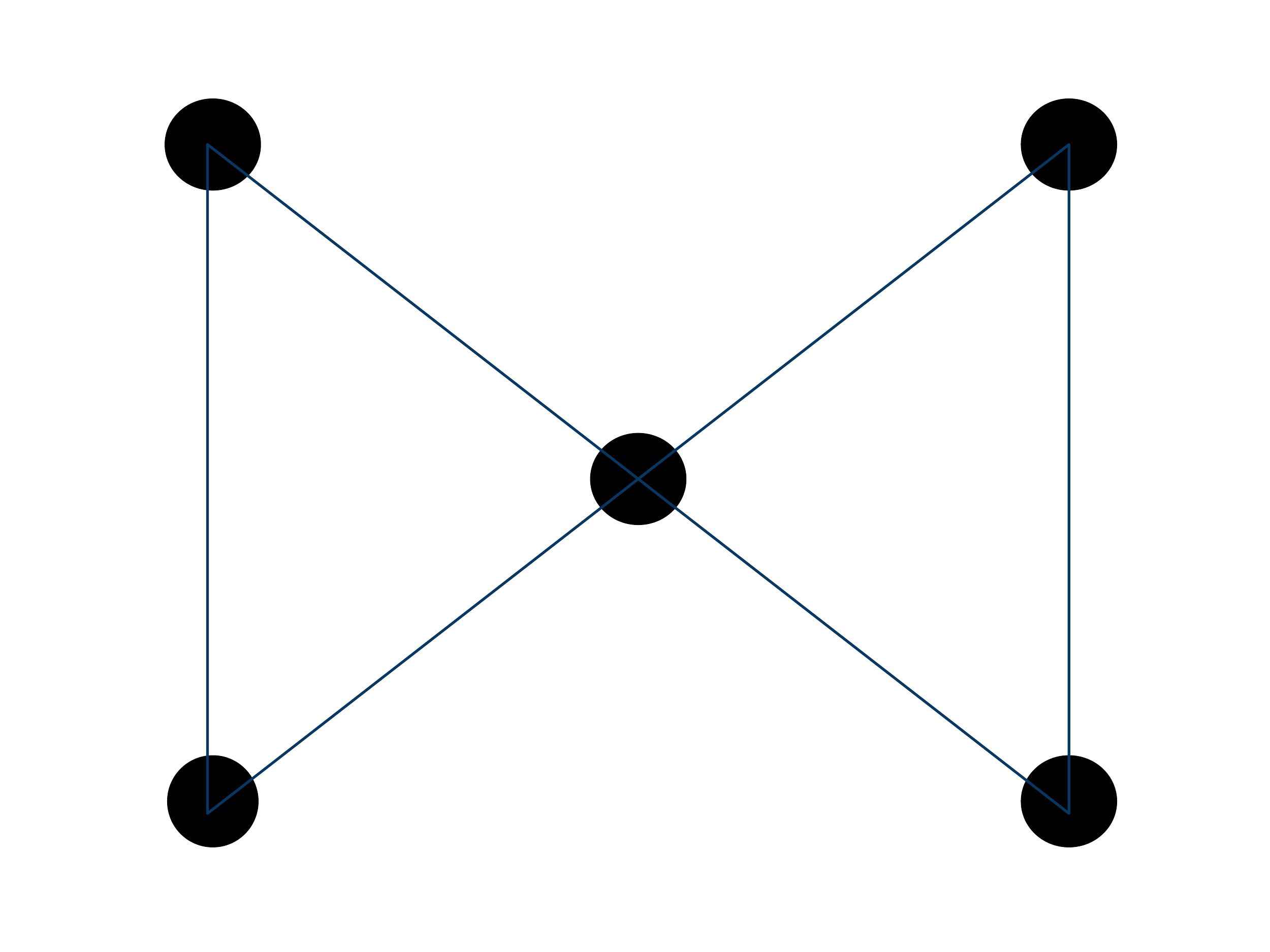}}
     \resizebox{50mm}{!} {\includegraphics[scale=.25]{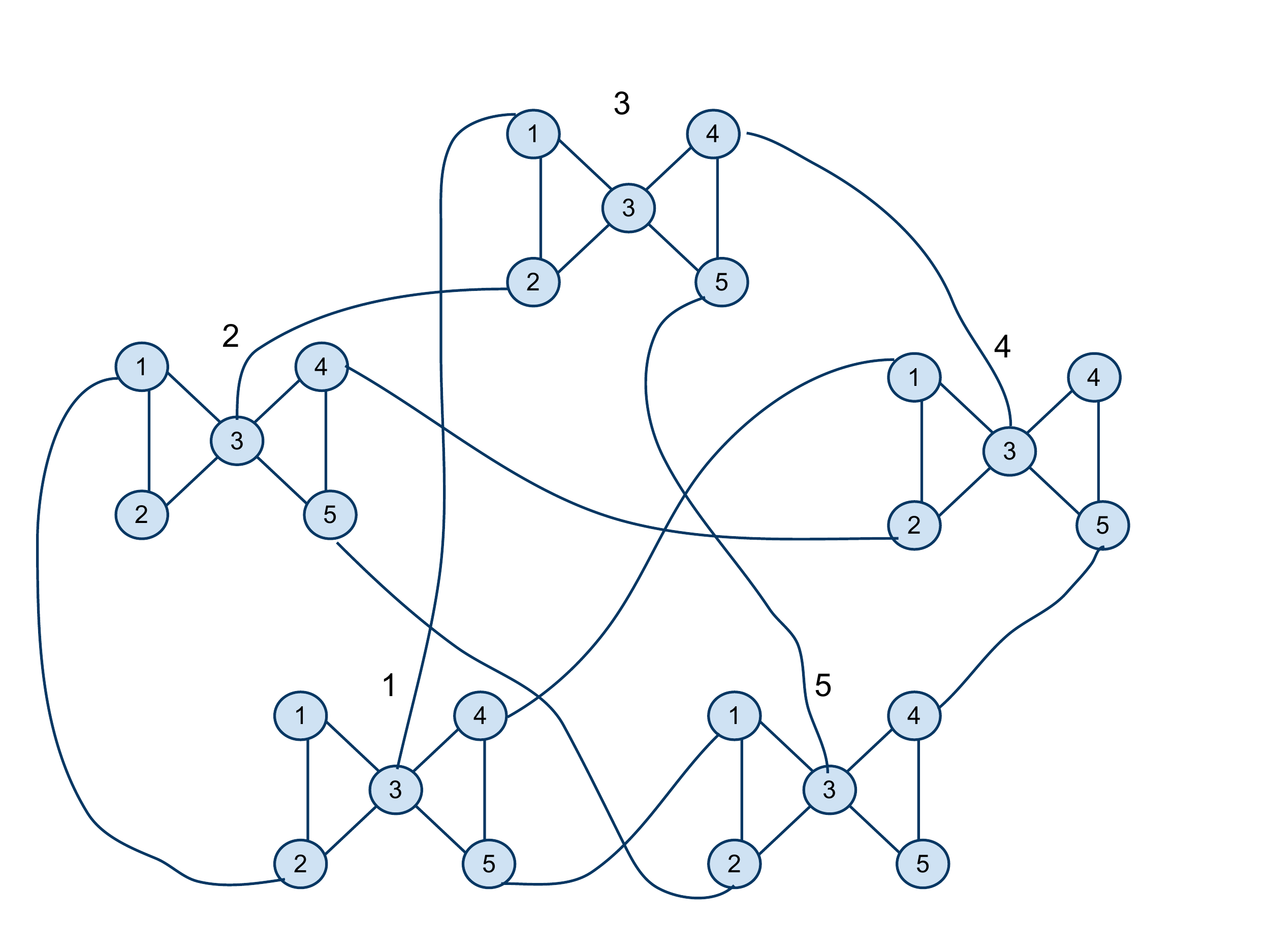}}
     \caption{ (a) Butterfly or Bowtie Graph $BF(3,3)$ (b) OTIS on $BF(3,3)$} 

   \end{center}
 \end{figure}
  
%
%

\subsection{Organization}
We have organized the paper in eight sections. In section 2, we give some preliminaries, in section 3, we give a brief outline of our work,
in section 4, we discuss the proof for Hamiltonicity on $OTIS(BF(2m+1,2n+1))$ and $OTIS(BF(2m+1,2k))$ thereby, proving that Hamiltonicity of base graph is not a necessary condition 
for the OTIS to be Hamiltonian. In section 5 we prove that  $OTIS(BF(4,4))$ and $OTIS(BF(4,6))$ are non-Hamiltonian, which proves that it 
is not sufficient for the base graph to have Hamiltonian path, for the OTIS to be Hamiltonian. In section 6, we show that OTIS network built
on the other class of Butterfly graphs, \cite{Barth} is Hamiltonian. We discuss our algorithm to create two Independent
Spanning Trees in time linear in the number of vertices in section 7. Finally, we conclude in section 8 mentioning some interesting open directions that need further exploration. 
%
%

\section{Preliminaries}
We will use standard graph theoretic terminology. Let $G = (V,E)$ be a finite undirected simple graph with vertex set $V(G)$ and edge set $E(G)$. 
For a vertex $v \in V(G)$, by $\func{deg}(v)$ we shall denote the degree of $v$ in $G$. The maximum degree among the vertices of $G$ is denoted by $\Delta(G)$ and the minimum degree by $\delta(G)$.
$\func{diam}(G)$ denote the diameter of $G$ and it is defined as the maximal distance between any two nodes in $G$. The connectivity of $G$, $\kappa(G)$ denotes the minimum number of vertices, which when removed, disconnects $G$.
The OTIS network denoted as $OTIS(G)$, derived from the base or basis or factor graph $G_B$, is a graph with vertex set:
\begin{align}
V(OTIS(G_B)) \define \set{\left\langle u,v\right\rangle | u,v \in V(G_B)},\notag
\end{align}
And edge set:
\begin{align}
E(OTIS(G_B)) \define \set{(\left\langle v,u\right\rangle, \left\langle v,u'\right\rangle) | v \in V(G_B), (u,u') \in E(G_B)} \cup \notag\\
\set{(\left\langle v,u\right\rangle, \left\langle u,v\right\rangle) | u, v \in V(G_B), u \neq v}.\notag
\end{align}
If the basis network $G_B$ has $n$ nodes, then $OTIS(G)$ is composed of $n$ node-disjoint subnetworks called clusters, each of which is isomorphic to $G_B$. We assume that the processor/nodes of the basis network is labeled $[n] = \set{1, \ldots, n}$, and the processor/node label $\langle g,u\rangle$ in $OTIS$ network $OTIS(G)$ identifies the node indexed $u$ in cluster $g$, and this corresponds to vertex $\langle g,u\rangle \in V(OTIS(G_B))$. Subsequently, we shall refer to $g$ as the cluster address of node $\langle g,u\rangle$ and $u$ as its processor address.\par

The vertices of the base graph $BF(m,n)$, [$m,n \in \mathbb{N}$]  of $OTIS(BF(m,n))$, is labeled with indices $\lbrace 1, 2, \ldots, c, c+1, \ldots, i \rbrace \subset \mathbb{N} $,
where $c$ denotes the label of the cutvertex, and $i$ denotes the label of the last vertex in the base graph and hence, $i= |V(G_B)|, m=c, n=i-c+1$. (Figure 2)
 
\begin{figure}[ht]
\begin{center}
\includegraphics[scale=.2]{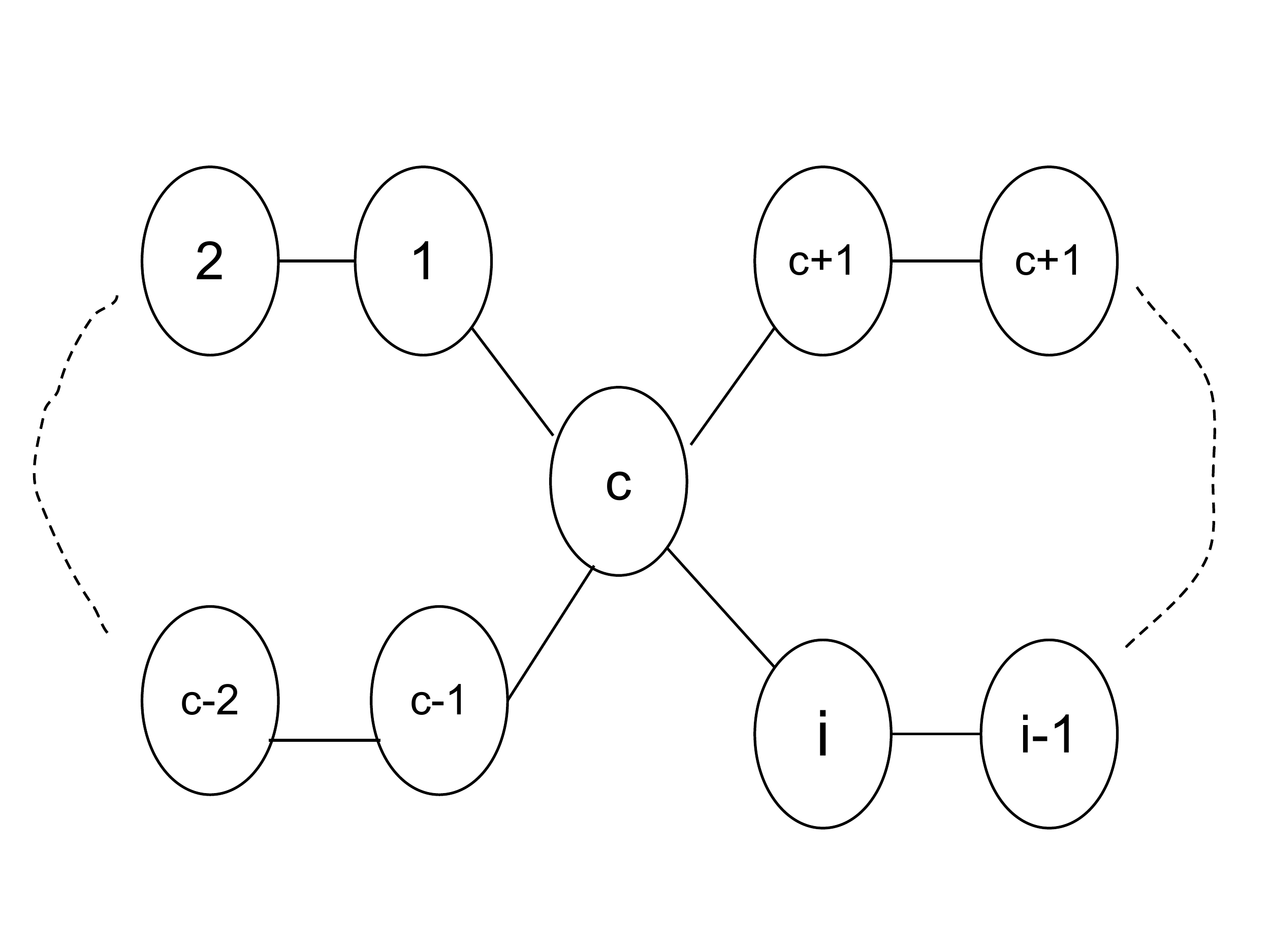} 
     \caption{Labeling $BF(m,n)$, where $i= |V(G_B)|, m=c, n=i-c+1$}

   \end{center}
 \end{figure}

To enhance readability, sometimes we will
mention a cluster $g$ and denote edges for which, both endpoints are within $g$, as $(x,y)$ which denote edges $(\langle g, x\rangle , \langle g, y \rangle)$.
A vertex is called "saturated" if its Hamiltonian neighbours, i.e, neighbours in a Hamiltonian Cycle, are explicitly identified.  

Based on the existing results following properties hold for $OTIS(G)$:
\begin{proposition}[\cite{Chen}]
\label{PROP1:SEC2}
Given basis graph $G = (V,E)$, with $\card{V} = n$, $\Delta(G) = \Delta$, $\delta(G) = \delta$, $\func{diam}(G) = d$,  and $\kappa(G) = k$, following holds for $OTIS(G)$:
\begin{enumerate}
\item	$\deg(\langle u,v\rangle) = \deg(v) + 1$ when $u \neq v$, and $\deg(v)$ otherwise.
\item	$\Delta(OTIS(G)) = \Delta + 1$.
\item	$\delta(OTIS(G)) = \delta$.
\item $\func{diam}(OTIS(G)) = 2d + 1$.
\end{enumerate}
\end{proposition}

%
%

\section{Outline of the Work}
We first investigate the Hamiltonicity of $OTIS(BF(2m+1,2n+1))$ and $OTIS(BF(2m+1,2k))$, where $ m, n, k \in \mathbb{N} $ and prove that both of them are Hamiltonian.
We give explicit constructions of Hamiltonian Cycles on these two classes. Thus we answer the question that, the base graph need not be Hamiltonian, for the OTIS-network to be Hamiltonian, as 
the generalized bowtie graphs, $BF(2m+1,2n+1)$ and $BF(2m+1,2k)$ are clearly Non-Hamiltonian. 

\begin{lemma}
Number of edge-disjoint Hamiltonian Cycles on a simple graph with minimum degree $\delta$ is at most $\left \lfloor{\frac{\delta}{2}}\right \rfloor$.
\end{lemma}

\begin{proof}
 Any vertex $v$ has $deg(v)$ number of edges incident on it. If possible, let there be  $H_i$ number of edge-disjoint Hamiltonian cycles on the graph. 
Each Hamiltonian Cycle will use exactly two of the $deg(v)$ edges incident on vertex $v$. Hence, $v$ can be included in at most $\frac{deg(v)}{2}$ Hamiltonian Cycles,
if $deg(v)$ is even, $\frac{deg(v)-1}{2}$ Hamiltonian Cycles, if $deg(v)$ is odd. Hence it is easily seen that $H_i$ is upperbounded by $\left \lfloor{\frac{\delta}{2}}\right \rfloor$.
\end{proof}

\noindent{The crucial observation that will be exploited for Hamiltonian Cycle construction on $OTIS(BF(2m+1,2n+1))$ and $OTIS(BF(2m+1,2k))$ is the following:}\\

\begin{observation}{There are only 4 kinds of vertex-degrees in the Bowtie-OTIS, $2, 3, 4, 5$ and the Bowtie-OTIS is 2-edge connected. Also there is exactly one vertex of degree $4$, namely $\langle c,c \rangle$, exactly
$(|V_B|-1)$ vertices of degree $2$ ($\langle x,x \rangle$ where $x \neq c$) and $(|V_B|-1)$ vertices of degree 5 ($\langle x,c \rangle$ where $\forall x \in (\lbrace 1, 2, \ldots, |V_B| \rbrace \backslash c)$) .  
Rest of the vertices are all of degree $3$.} 
\end{observation}

The correctness of this observation follows from Proposition 3.1.\\


%
%
%

Using Lemma 3.1 and Observation 3.1, it is easily seen the number of edge-disjoint Hamiltonian Cycles on $OTIS(BF(2m+1,2n+1))$ and $OTIS(BF(2m+1,2k))$ can be at most $\left \lfloor{\frac{2}{2}}\right \rfloor = 1$.
We give construction for Hamiltonian Cycle and discuss how these constructions can be used to generate two Independent Spanning Trees on
$OTIS(BF(2m+1,2n+1))$ and $OTIS(BF(2m+1,2k))$ in time linear in the number of vertices of the OTIS-network (in section 7).\par

Next we address the question whether it is sufficient for the base graph to have Hamiltonian path, for the OTIS to be Hamiltonian. We answer this in negative, proving that the $OTIS(BF(4,6))$ and $OTIS(BF(4,4))$
are both Non-Hamiltonian. It is easy to see that the base graph, in both the cases, admits Hamiltonian Path. \par

Lastly, we consider the the OTIS network built of butterfly graph mentioned in \cite{Barth}. 

\section{Proof of Hamiltonicity of $OTIS(BF(2m+1,2n+1))$ and $OTIS(BF(2m+1,2k))$}

We give constructive proofs for both $OTIS(BF(2m+1,2n+1))$ and $OTIS(BF(2m+1,2k))$. First we state two Inference Rules that will be used to construct Hamiltonian Cycles.

\begin{itemize}
 \item[IR 1:] {If a vertex of degree $\geq3$, gets saturated, the rest of its edges, not used in the saturation, becomes Non-Hamiltonian edges and are deleted from the graph.}
 \item[IR 2:]{If $(\langle g_1, u \rangle, \langle g_2, v \rangle)$ is an edge between the vertices  $\langle g_1, u \rangle$ and $\langle g_2, v \rangle$, both of degree
$3$, and  if the edge $ (\langle g_1, u \rangle, \langle g_2, v \rangle)$ is identified as Non-Hamiltonian, then all other edges incident to the vertices  $\langle g_1, u \rangle$ and $\langle g_2, v \rangle$
are forced to be Hamiltonian.\\[.2cm]
Once the edge $ (\langle g_1, u \rangle, \langle g_2, v \rangle)$ is identified as Non-Hamiltonian, it is dropped from the potential set of edges required to construct
Hamiltonian cycle. So now, exactly 2 potential edges are incident to each $ \langle g_1, u \rangle$ and $\langle g_2, v \rangle$ and hence are forced to be Hamiltonian edges.}

\end{itemize}

The steps in the construction are as follows:\\

\begin{itemize}

\item[Step 1:]{We identify the key Non-Hamiltonian edges whose endpoints lies within the same cluster, explicitly and delete them.}
\item[Step 2:]{In this process some vertices becomes saturated; we apply IR 1 on these vertices.}
\item[Step 3:]{The previous step, in turn decides Hamiltonian edges of the remaining vertices(due to IR 2).}  
 
\end{itemize}

\begin{observation}
The constructions can be implemented as algorithm to construct Hamiltonian Cycles on $OTIS(BF(2m+1,2n+1))$ and $(OTIS(BF(2m+1,2k))$ in time $O(m|V_B|)$, i.e., 
in time linear in the number of vertices of the OTIS graph. [$|V_B| = $ number of clusters].  
\end{observation}

This observation follows from the fact that, the number of intracluster edges, deleted per cluster, in these constructions, is of $O(m)$, assuming $m>n, m>k$, without any loss of generality.
In Step 1 of the construction, non-Hamiltonian intracluster edges are explicitly identified for all the clusters. Therefore, this step takes time, proportional to 
the number of clusters, i.e., $O(|V_B|)$. Hence the Hamiltonian Cycle construction takes time $O(m|V_B|)$, i.e., in time linear in the number of vertices in the base graph.

\subsection{Hamiltonicity of $OTIS(BF(2m+1,2n+1))$}
We identify the key non-Hamiltonian edges (Step 1 of the construction) in three parts.
First we identify the key Non-Hamiltonian edges for $OTIS(BF(3,2n+1))$, $n>1$ . Then identify the key Non-Hamiltonian edges for $OTIS(BF(2m+1,2m+1))$, $m>1$. \footnote{We give explicit construction for $m=1$} 
Lastly, we identify the key Non-Hamiltonian edges for any $OTIS(BF(2m+1,2n+1))$, where $m >1$, $n>3$ and $n>m$.\footnote{We give explicit construction for $m=2, n=3$} This completes the proof that any $OTIS(BF(2m+1,2n+1))$ [$m,n \in \mathbb{N}$] is Hamiltonian..
Note that, in these computations, the label $0$ is same as label $c$.

Here we show the construction of $OTIS(BF(2m+1,2m+1))$, $m>1$ and argue its correctness. The constructions for $OTIS(BF(3,2n+1))$, $n>1$ and $OTIS(BF(2m+1,2n+1))$, where $m >1$, $n>3$ and $n>m$
and for $OTIS(BF(3,3))$ and $OTIS(BF(5,7))$ is given in appendix A.


\subsubsection{Key non-Hamiltonian edges for $OTIS(BF(2m+1,2m+1))$, $m>1$.}
We determine the key non-Hamiltonian intracluster edges for each cluster. 
\begin{description}
\item[Cluster 1:]
The sets $S_1= \lbrace (2,3), (4,5), \ldots, (c,c-1) \rbrace$, $S_2 =\lbrace (c+2,c+3), (c+4,c+5), \ldots, (i-2,i-1)\rbrace$ and $(c,i), (c,c-1)$ and $(c,c+1)$.
\item[Cluster $(c+1)$:]
The sets $S_1= \lbrace (2,3), (4,5), \ldots, (c,c-1) \rbrace$, $S_2 =\lbrace (c+2,c+3), (c+4,c+5), \ldots, (i-2,i-1)\rbrace$ and $(c,i), (c,c-1) $and $(c,1)$.
\item[Cluster $(c-1)$:]
The set $S_3= \lbrace (2,3), (4,5), \ldots, (c-3,c-2) \rbrace$, $(c,c+1), (c,1)$ and $(c+3,c+4)$ iff $(c+4) \neq i$.
\item[Cluster $(c-2)$:]
The set $S_4= \lbrace (1,2), (3,4), \ldots, (c-1,c) \rbrace$, $(c,c+1)$ and $(c+3,c+4)$ iff $(c+4) \neq i$.
\item[Cluster $i$:]
The set $S_5= \lbrace (c+2,c+3), \ldots, (c-3,c-2) \rbrace$, $(c,c+1), (c,1)$ and $(3,4)$ iff $(c-1) \neq 4$.
\item[Cluster $(i-1)$:]
The set $S_6= \lbrace (c+1,c+2), \ldots, (c,i) \rbrace$, $(c,1)$ and $(3,4)$ iff $(c-1) \neq 4$.
\end{description}

For Clusters $\lbrace 2, 4, 6, \ldots (c-1)\rbrace$ and $\lbrace (c+2), (c+4), \ldots, i \rbrace$, delete edges $(c,1)$ and $(c,c+1)$.\\
For Clusters $\lbrace 1, 3, 5, \ldots (c-2)\rbrace$ delete edges $(c,c-1)$ and $(c,c+1)$.\\
For Clusters $\lbrace (c+1), (c+3), \ldots (i-1)\rbrace$ delete edges $(c,1)$ and $(c,i)$.\\

Also $\forall$ cluster $x, 1 \leq x \leq i$, delete edges $(x-2, x-1)$ and $(x+1, x+2)$.\\
%

\subsubsection{Correctness Argument for Hamiltonicity for $OTIS(BF(2m+1,2m+1))$, $m>1$.}

\begin{claim}
 All the Hamiltonian edges of $OTIS(BF(2m+1,2m+1))$ can be inferred by deleting the Key edges mentioned and using the inference rules IR 1 and IR 2.
\end{claim}

\begin{proof}
First we concentrate on the clusters $\lbrace 2, 4, 6, \ldots (c-1)\rbrace$.
\begin{enumerate}
\item{We mark the intercluster edges $(\langle 1, p \rangle, \langle p, 1 \rangle)$, $(\langle c-1, p \rangle, \langle p, c-1 \rangle)$ and 
$(\langle c-2, p \rangle, \langle p, c-2 \rangle)$ as Hamiltonian edges (Using IR 1) $ \forall p \in \lbrace 2, 4, 6, \ldots (c-1)\rbrace$.}

\item{In clusters $\lbrace 3, 5, 7, \ldots (c-2)\rbrace$, applying IR 2 for the vertex 2, we infer that the intercluster edges $(\langle 2, p \rangle, \langle p, 2 \rangle)$
$ \forall p \in \lbrace 3, 5, 7, \ldots (c-2)\rbrace$ are Hamiltonian edges.}

\item{We also know that $(\langle p, x-2 \rangle, \langle x-2, p \rangle)$, $(\langle p, x-1 \rangle, \langle x-1, p \rangle)$, $(\langle p, x+2 \rangle, \langle x+2, p \rangle)$,
$(\langle p, x+1 \rangle, \langle x+1, p \rangle)$ and the edges $(\langle p, x-3 \rangle, \langle p, x-2 \rangle)$, $(\langle p, x+3 \rangle, \langle p, x+2 \rangle)$  are Hamiltonian edges (Using IR 2) $ \forall p \in \lbrace 2, 4, 6, \ldots (c-1)\rbrace$.}

\item {Using (2) and (3) and IR 2, we infer set of non-Hamiltonian edges in clusters $\lbrace 2, 4, 6, \ldots (c-1)\rbrace$: \\

\begin{itemize}
 \item {The set $S_{e1} = \lbrace (x-2,x-1), (x-4,x-3), \ldots (2,3) \rbrace$ when $(x-2) \neq c$. Else ignore this set.\footnote{For Cluster 2, this set is ignored.}}
\item{The set $S_{e2} = \lbrace (x+1,x+2), (x+3,x+4), \ldots (c-2,c-1) \rbrace$ when $(x+1) \neq c$. Else ignore this set. \footnote{For Cluster $(c-1)$, this set is ignored.}}
\end{itemize}}

\end{enumerate}

This completes the description of non-Hamiltonian edges within the clusters $\lbrace 2, 4, 6, \ldots (c-1)\rbrace$, which decides all the Hamiltonian neighbours of the vertices within these clusters. Below we illustrate with the example of Cluster 2.
 \begin{itemize}
  \item {Hamneighbour$\langle2,1\rangle = \langle2,2\rangle, \langle1,2\rangle$.}
   \item {Hamneighbour$\langle2,2\rangle= \langle2,1\rangle, \langle1,3\rangle$.}
 \item {Hamneighbour$\langle2,3\rangle = \langle2,2\rangle, \langle3,2\rangle$.}
 \item {Hamneighbour$\langle2,4\rangle = \langle2,5\rangle, \langle4,2\rangle$.}
 \item {Hamneighbour$\langle2,5\rangle = \langle2,4\rangle, \langle5,2\rangle$.}

$\vdots$

 \item {Hamneighbour$\langle2,(c-3)\rangle = \langle2,(c-2)\rangle, \langle(c-3),2\rangle$.}
\item {Hamneighbour$\langle2,(c+3)\rangle = \langle2,(c+2)\rangle, \langle(2,(c+4)\rangle$.}

$\vdots$

\item {Hamneighbour$\langle2,(i-1)\rangle = \langle2,(i-2)\rangle, \langle(2,i\rangle$.}
\item {Hamneighbour$\langle2,i\rangle = \langle2,(i-1)\rangle, \langle(2,c\rangle$.}
\end{itemize}

By similar arguments, we infer set of non-Hamiltonian edges in clusters $\lbrace 1, 3, 5, \ldots (c-2)\rbrace$, which are as follows:\\
\begin{itemize}
 \item {The set $S_{o1} = \lbrace (x-2,x-1), (x-4,x-3), \ldots (1,2) \rbrace$ when $(x-2) \neq c$. Else ignore this set.\footnote{For Cluster 1, this set is ignored.}}
\item{The set $S_{o2} = \lbrace (x+1,x+2), (x+3,x+4), \ldots (c-3,c-2) \rbrace$ when $(x+1) \neq c$. Else ignore this set.}
\end{itemize}
This completes the description of non-Hamiltonian edges within the clusters $\lbrace 1, 3, 5, \ldots (c-2)\rbrace$. \\

By symmetry, we can infer the set of non-Hamiltonian edges in clusters $\lbrace (c+2), (c+4), \ldots, i \rbrace$ and $\lbrace (c+1), (c+3), \ldots (i-1)\rbrace$.\\
\begin{figure}[ht]
\begin{center}
\includegraphics[scale=.23]{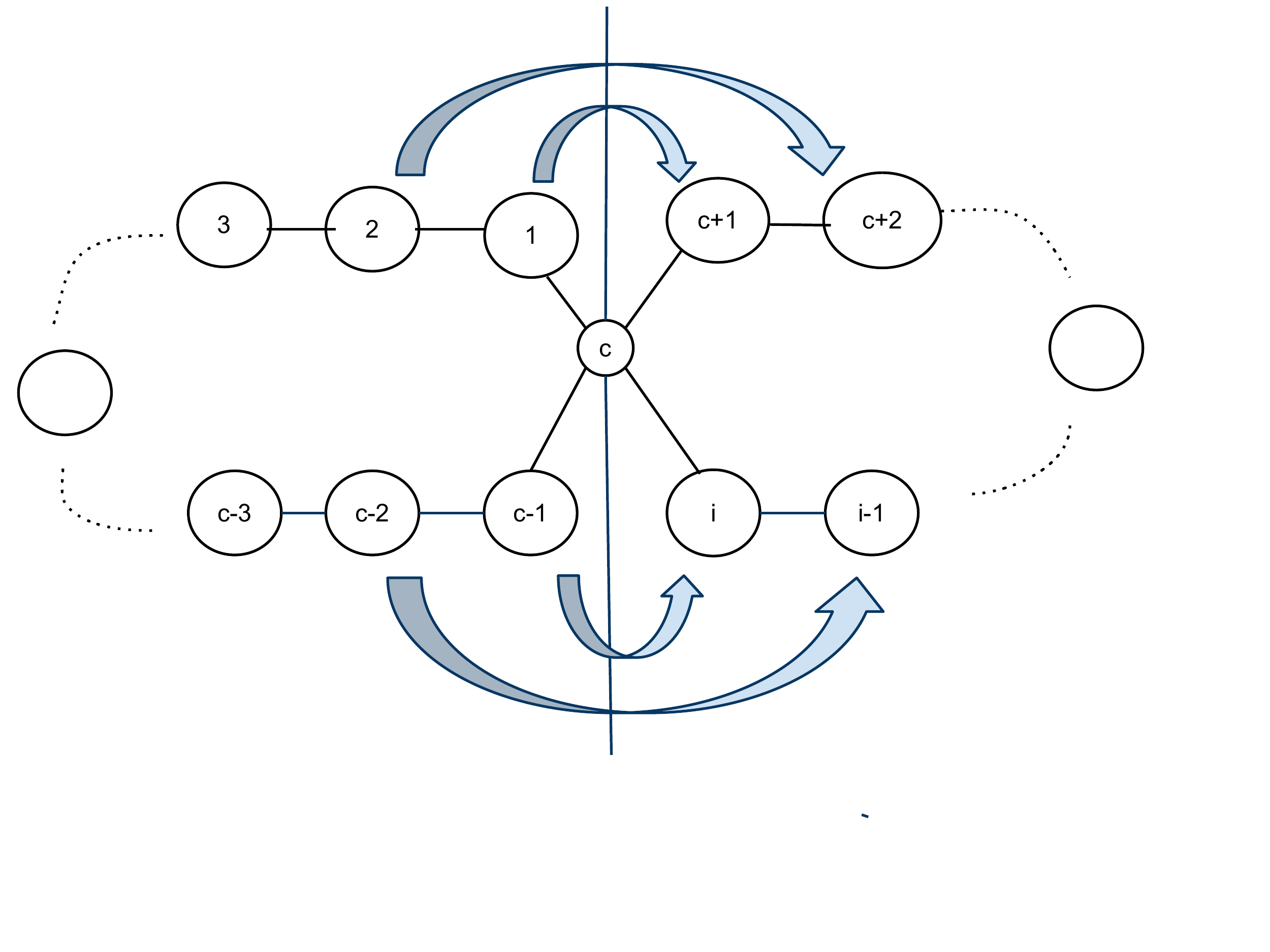}
\caption{$BF(2m+1,2m+1)$: The mapping between the labels}
\end{center}
\end{figure}

Hence, all the Hamiltonian edges of $OTIS(BF(2m+1,2m+1))$ can be inferred using be deleting the Key edges mentioned and using the inference rules IR 1 and IR 2.
\end{proof}



%

%
%

\subsection{Construction of Hamiltonian Cycle for $OTIS(BF(2m+1,2k))$}

Here we identify the key non-Hamiltonian edges in two parts.

First we identify the key Non-Hamiltonian edges for $OTIS(BF(3,2k))$ and then for $OTIS(BF(2m+1,2k))$, where $m>1$. 
\subsection{Key Non-Hamiltonian edges for  $OTIS(BF(3,2k))$}
\begin{description}
 \item [Cluster 1:]
$(3,2), (3,4), (3,i)$
\item[Cluster 2:]
$(3,1), (3,4)$ And the set $S_2=\lbrace (5,6), (7,8), \ldots, (i-1,i)\rbrace$
\item[Cluster 3:]
$(3,1), (3,4)$
\item[Cluster 4:]
$(3,1), (3,2), (3,i)$
\item[Cluster $5, 6, \ldots, (i-1)$:]
$(3,2), (3,i)$
\item[Cluster $i$:]
$(3,1), (3,2)$ and the set $S_i=\lbrace (4,5), (6,7), \ldots, (i-2,i-1)\rbrace$ 
\end{description}
Now join the intercluster edges at the at both endvertices of the deleted intercluster edges. This completes the Hamiltonian Cycle.

The construction for $OTIS(BF(2m+1,2k))$, where $m>1$ is shown in appendix B.

%
%

\section{Proof of Non-Hamiltonicity of $OTIS(BF(4,4))$ and  $OTIS(BF(4,6))$}

\subsection{ Proof that $OTIS(BF(4,4))$ is not Hamiltonian}

noindent{We present the proof of non-Hamiltonicity using counting argument.\\
\noindent{First, let us count the total number of edges in the graph. Total number of edges in $OTIS(BF(4,4))$ is $\frac{\displaystyle\sum_{i = 1}^{n} {d_i} }{2} =77$, where $d_i$ denotes degree of vertex $i$
and $n= |V(OTIS(BF(4,4)))|$. If a Hamiltonian Cycle exists, it will use up $49$ edges, as $n = 49$. So there are exactly $(77-49)=28$ non-Hamiltonian edges in the graph.

Now, let us count the number of non-Hamiltonian edges in a different way. Let us look into $OTIS(BF(4,4))$ carefully.
There are six degree 5 vertices, namely, $<1,4>, <2,4>, <3,4>, <5,4>, <6,4>$ and $ <7,4> $. Neighbours of these six vertices are disjoint. Also there is exactly one, degree 4 vertex: $<4,4>$
These vertices together will contribute to $6 \cdot 3 +2 = 20$ non-Hamiltonian edges.
Now let us look into the subgraph induced by vertices of degree 3 only, which do not have any degree 5 or degree 4 neighbour (Figure 4). Maximum Independent Subset induced by these vertices is of cardinality = 9. Hence, these vertices, accounts for 9 non-Hamiltonian edges which have not been counted yet. Hence, the total number of
non-Hamiltonian edges $= 20+9 =29$, which does not agree with the previous count, $28$.  Hence a contradiction. So $OTIS(BF(4,4))$ is not Hamiltonian.}

\begin{figure}[ht]
\begin{center}
\includegraphics[scale=.28]{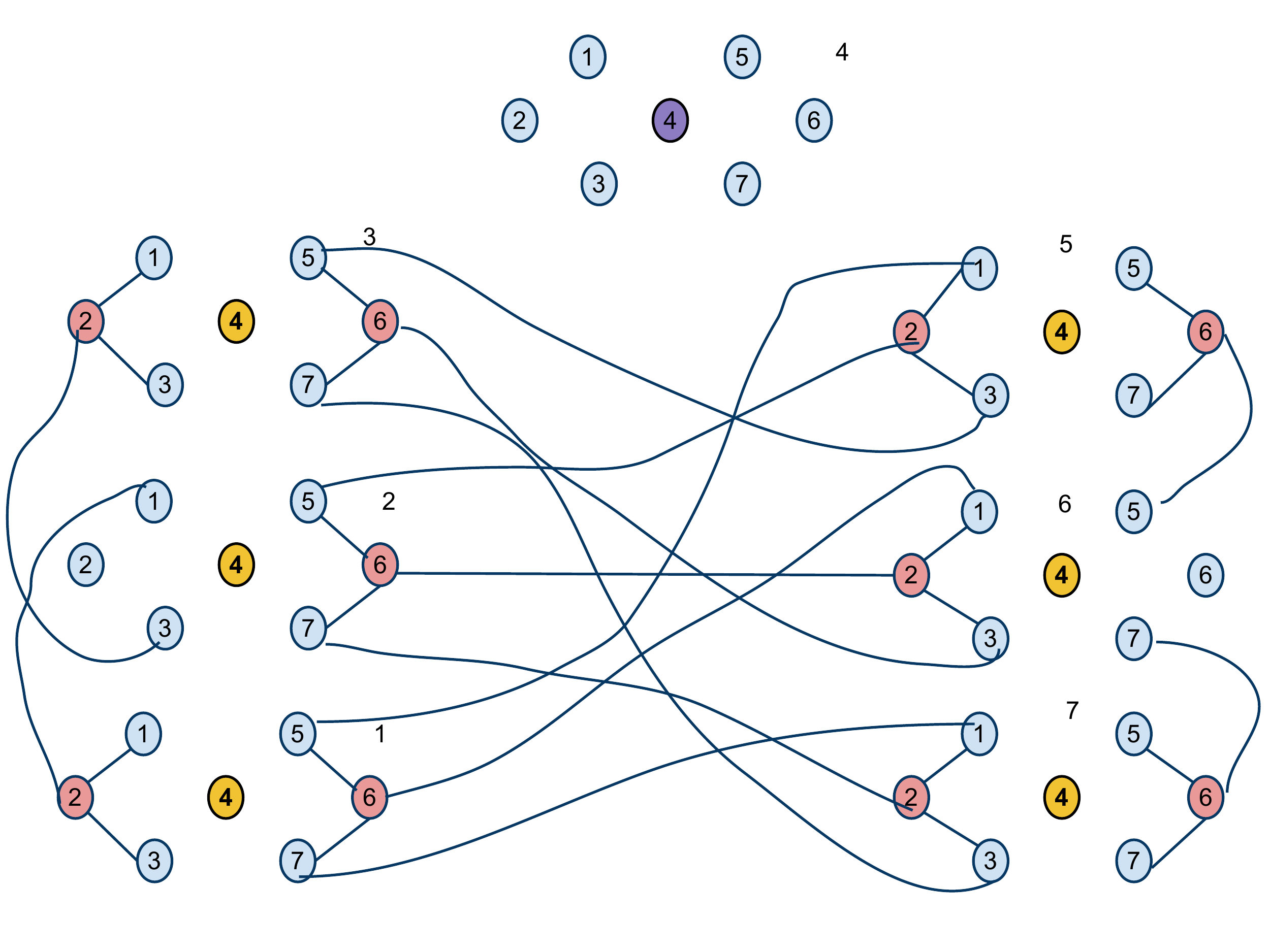}
\caption{$OTIS(BF(4,4))$:Vertices of degree 3 vertices, with no degree 5 and 4 neighbours are shown in red.}
\end{center}
\end{figure}

However, this argument cannot be used to prove that $OTIS(BF(4,6))$ is Non-Hamiltonian, as $OTIS(BF(4,6))$ has a cycle cover of length 2. Hence, its Non-Hamiltonicity
cannot be captured through this counting argument. The proof of Non-Hamiltonicity of $OTIS(BF(4,6))$ is presented in appendix C.

%
%

\section{ $OTIS(BF(n))$ is Hamiltonian}
It has been shown in the paper\cite{Barth} that $BF(n)$ has two edge-disjoint Hamiltonian Cycles, by giving a recursive method of construction of the cycles. 
Hence the base network of $OTIS(BF(n))$ is Hamiltonian. Combining this result, with \cite{Parhami}, it is easily seen that $OTIS(BF(n))$ is Hamiltonian. 

%
%

\section{Independent Spanning Trees Construction}

\begin{theorem}[\cite{Itai}]
Given any 2-connected graph G and a vertex $s$ in G, there are two spanning trees such that the paths from $s$ to any other node in G on the trees are node disjoint.
\end{theorem}
\noindent{The graph families we are considering, $OTIS(BF(2m+1,2n+1))$ and $OTIS(BF(2m+1,2k))$, are 2-edge connected. Hence there exists two Independent Spanning Trees. We construct
two Independent Spanning Trees as follows:\\ 
We know that construction of Hamiltonian Cycle is linear in number of vertics of the OTIS-network on $OTIS(BF(2m+1,2n+1))$ and $OTIS(BF(2m+1,2k))$ (From observation 4.1).}

\noindent{Once the Hamiltonian Cycle is constructed, we can construct two rooted Independent Spanning Trees in $O(1)$ time as follows:\\}
\begin{itemize}
 \item {Pick any vertex as root and denote it as $r$}
 \item{Delete one edge incident to $r$. This gives a spanning Tree $T_1$}
 \item{Now, retain the edge previously deleted, and delete the other edge incident to $r$. This gives another spanning tree $T_2$}
 
\end{itemize}

Clearly, for all vertices $v$ in the graph, paths connecting $r$ and $v$ in $T_1$ and $T_2$ are edge and vertex disjoint. 
Even if we make $BF(2m+1,2n+1)$ or $BF(2m+1,2k)$ denser by introducing chords
inside the cycles $C_{2m+1}$, $C_{2n+1}$ or $C_{2k}$ such that at least one of the vertices retain degree two, the two rooted Independent Spanning Trees constructed by
this algorithm in will still be valid, which will run in time linear in the number of vertices of the OTIS-graph. But the generalized algorithm will show poor performance
as it runs in time linear in the number of edges.

\section{Conclusion}
In this paper, we have shown that existence of Hamiltonian Path on the base graph is not 
a sufficient condition for the OTIS network to be Hamiltonian. It would be interesting to investigate whether it is a necessary condition.
Another important open direction is to see if the Independent Spanning Tree conjecture holds on any $k$-connected OTIS network for arbitrary values of $k$, as this is
a most important aspect of fault tolerance is any distributed network.

\bibliographystyle{splncs}
\bibliography{ref}

\appendix
\section{Constructions for $OTIS(BF(c=3,2n+1))$ and $OTIS(BF(2m+1,2n+1))$, where $m >1$, $n>3$ and $n>m$}

\subsection{Key non-Hamiltonian edges for  $OTIS(BF(c=3,2n+1))$}
We determine the key non-Hamiltonian intracluster edges for each cluster. 
\begin{description}
\item[Cluster 1:]
The set $S_1 =\lbrace (c+2,c+3), (c+4,c+5), \ldots, (i-2,i-1)\rbrace$ and $(c,i), (c,c-1)$ and $(c,c+1)$.
\item[Cluster 2:]
$(c,1)$ and $(c, c+1)$ and the set $S_2=\lbrace(6, 7), (i-3, i-2)\rbrace$ iff $7> i$, else ignore this set.
\item[Cluster 3:]
$(c,1), (c,c+1)$.
\item[Cluster $(c+1)$:]
$(c,1), (c, c-1), (c,i)$ and $(c+2, c+3), (i-2, i-1)$
\item[Cluster $(i-1)$:]
$(c,1), (c,i)$ and the set $S_{(i-1)} = \lbrace (4,5), (6, 7), (i-3, i-2) \rbrace $ iff $5<(i-4)$, else ignore this set.
\end{description}

Also $\forall$ cluster $x, 1 \leq x \leq i$, delete edges $(x-2, x-1)$ and $(x+1, x+2)$.\\

\subsection{Key non-Hamiltonian edges for $OTIS(BF(2m+1,2n+1))$, where $m >1$, $n>3$ and $n>m$ }

\begin{description}
\item[Cluster 1:]
$ (c,c-1), (c,c+1), (c,i)$ and the sets $S_{1^l} =\lbrace (2,3), (4,5), \ldots, (c-1,c)\rbrace$. $S_{1^r} =\lbrace (c+2,c+3), (c+4,c+5), \ldots, (i-2,i-1)\rbrace$.
\item[Cluster 2:]
$(c,1), (c,i), (c-2,c-1)$and the set $S_2=\lbrace (c+5,c+6), (c+7,c+8), \ldots, (i-5,i-4)\rbrace$ where $ (i-5)\geq(c+5)$. Else ignore the set $S_2$.
\item[Cluster 3:]
$(i-1,i-2), (c,c-1), (c,c+1)$ and the set $S_3=\lbrace (c+5,c+6), (c+7,c+8), \ldots, (i-5,i-4)\rbrace$ where $ (i-5)\geq(c+5)$. Else ignore the set $S_3$.
\item[Cluster $4, \ldots, (c-3)$:]
$(c,1), (c,c+1), (i-2, i-1)$
\item[Cluster $(c-2)$:]
$(c,c-1), (c,c+1), (i-2,i-1)$.
\item[Cluster $(c-1)$:]
$(c,1), (c,i), (2,3)$ and the set $S_6 =\lbrace (c+4,c+5), (c+6,c+7), \ldots, (i-2,i-1)\rbrace$

\item[Cluster $c$:]
$(c,1), (c,c+1)$.

\item[Cluster $(c+1)$:]
$(c,1), (c,c-1), (c,i)$ and $(c+2,c+3), (i-2,i-1)$.
 
\item[Cluster $(c+2)$:]
$(c,c-1), (c,c+1), (i-1,i)$.
\item[Cluster $(c+3)$:]
$(c,1), (c,c+1), (i-1,i)$.
\item[Cluster $(c+4)$:]
$(c,c-1), (c,1), (i-1,i)$.

\item[Cluster $(c+5)$ to $(i-4)$:]
$(2,3)$ [only where $(i-5)\geq(c+5)$, else do not delete this edge.] $(c,1), (c,c-1)$ and $(i,i-1)$.
\item[Cluster $(i-3)$:]
$(c,1), (c,c-1), (i-1,i)$.
\item[Cluster $(i-2)$:]
$S_{(i-2)} = \lbrace (3,4),(5,6), \ldots, (c-2,c-1) \rbrace$ and $(c,1), (c,c+1)$ and $(i,i-1)$.
\item[Cluster $(i-1)$:]
The set $S_{{(i-1)}^l} = \lbrace (3,4),(5,6), \ldots, (c-2,c-1)$ and $ (c,1), (c,i)$ and the set $S_{{(i-1)}^r} =\lbrace (c+1,c+2), (c+3,c+4), \ldots, (i-3,i-2)\rbrace$.
\item[Cluster $i$:]
$(1,2), (c,c-1), (c,c+1)$ and the set $S_i =\lbrace (c+2,c+3), (c+4,c+5), \ldots, (i-2,i-1)\rbrace$.
\end{description}

In addition to this, $\forall$ cluster $x, 1 \leq x \leq (c-1)$, delete edges $(x-2, x-1)$ and $(x+1, x+2)$.\\

\subsection{Explicit constructions for $OTIS(BF(3,3))$ and $OTIS(BF(5,7))$} 
 We show the explicit constructions for $OTIS(BF(3,3))$  (Figure 5) and $OTIS(BF(5,7))$  (Figure 6).

\begin{figure}
   \begin{center}
     \includegraphics[scale=.3]{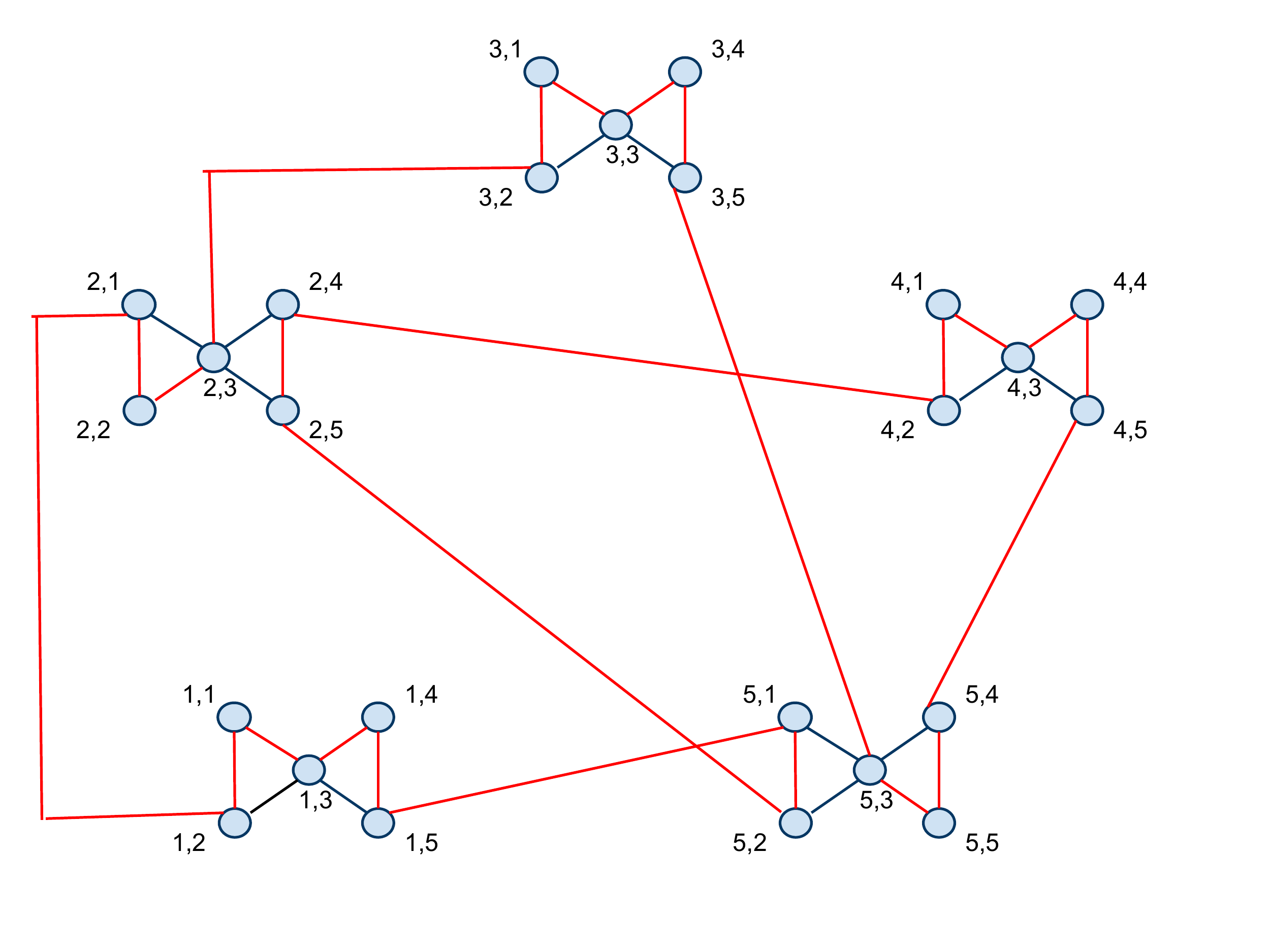}
 
     \caption{Hamiltonian Cycle on $OTIS(BF(3,3))$ shown in red colour}

   \end{center}
 \end{figure}

\begin{figure}
   \begin{center}
  \includegraphics[scale=.35]{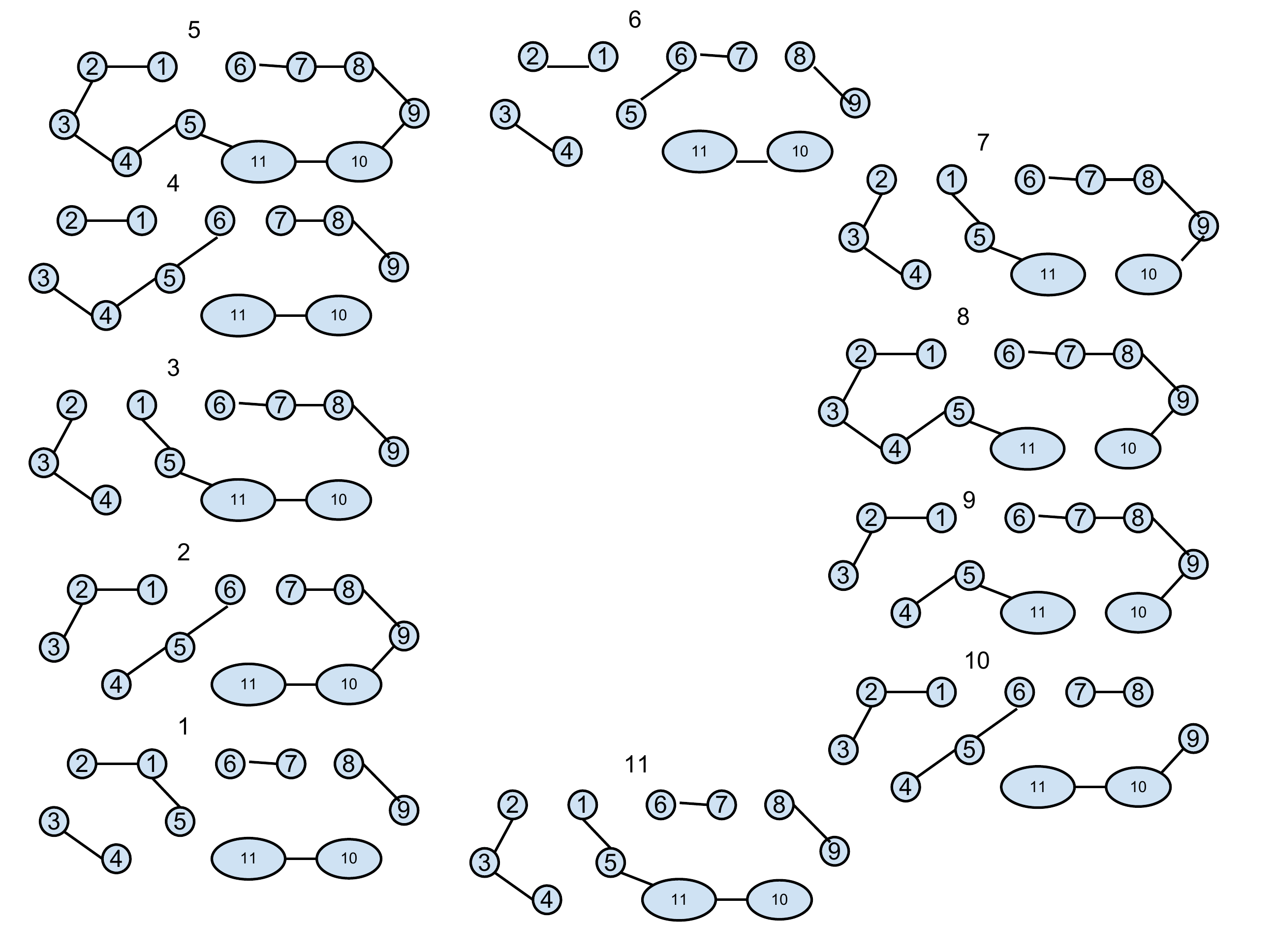}

     \caption{Joining the intercluster edges incident on the unsaturated vertices completes the Hamiltonian Cycle on $OTIS(BF(5,7))$}

   \end{center}
 \end{figure}

\section{Constructions for $OTIS(BF(3,2k))$ and $OTIS(BF(2m+1,2k))$, where $m>1$ }

\subsection{Key Non-Hamiltonian edges for $OTIS(BF(2m+1,2k))$, where $m>1$}

\begin{description}
 \item [Cluster 1:]
$(c,c-1), (c,c+1), (c,i)$ and the set $S_1 = \lbrace (2,3), (4,5), (6,7), \ldots, (c-1,c) \rbrace$. 

\item[Cluster 2:]
$(c,1), (c,c+1), (c-2,c-1)$ and the set $_2=\lbrace (c+2,c+3), (c+4,c+5), \ldots, (i-3,i-2)\rbrace$ where $ i\geq(c+3)$. Else ignore the set $S_2$.
\item[Cluster 3:]
$(c,c-1), (c,c+1)$ and the set $S_3=\lbrace (c+2,c+3), (c+4,c+5), \ldots, (i-3,i-2)\rbrace$ where $ i\geq(c+3)$. Else ignore the set $S_3$. Delete $(i-1,i)$ if $4<(c-1)$.
 
\item[Cluster $3, 4, \ldots, (c-3)$:]
$(i-1,i)$ if $4<(c-1)$
\item[Cluster $4, 5, \ldots, (c-3)$:]
$(c,1), (c,c+1)$ if $4<(c-3)$
\item[Cluster $(c-2)$:]
$(c,c-1), (c,c+1)$ 
\item[Cluster $(c-1)$:]
$(c,1), (c,i), (2,3)$ and the set $S_{(c-1)} =\lbrace (c+1,c+2), (c+3,c+4), \ldots, (i-2,i-1)\rbrace$
\item[Cluster $c$:]
$(c,1), (c,c+1)$.

\item[Cluster $(c+1)$:]
$(c,1), (c,c-1), (c,i)$ and the set $S_{(c+1)}= \lbrace (2,3), (4,5), \ldots, (c-3,c-2)\rbrace$

\item[Cluster $(c+2)$ to $(i-2)$:]
$(c,c-1), (c,i)$ and the edge $(2,3)$ if $(c+3) < i$.

\item[Cluster $(i-1)$:]
$ (c,c-1), (c,i)$ and the set $S_{i-1} =\lbrace (3,4), (5,6), \ldots, (c-4,c-3)\rbrace$ if $4<(c-1)$
\item[Cluster $i$:]
$(c,1), (c,c-1)$ and the sets $S_{{i}^l} =\lbrace (3,4), (5,6), \ldots, (c-4,c-3)\rbrace$ and $S_{{i}^r}= \lbrace (c+1,c+2), (c+3,c+4), \ldots, (i-2,i-1)\rbrace$
\end{description}
In addition to this, $\forall$ cluster $x, 1 \leq x \leq (c-1)$, delete edges $(x-2, x-1)$ and $(x+1, x+2)$.\\

\section{Proof that $OTIS(BF(4,6))$ is not Hamiltonian}
\begin{figure}[ht]
\begin{center}
\includegraphics[scale=.28]{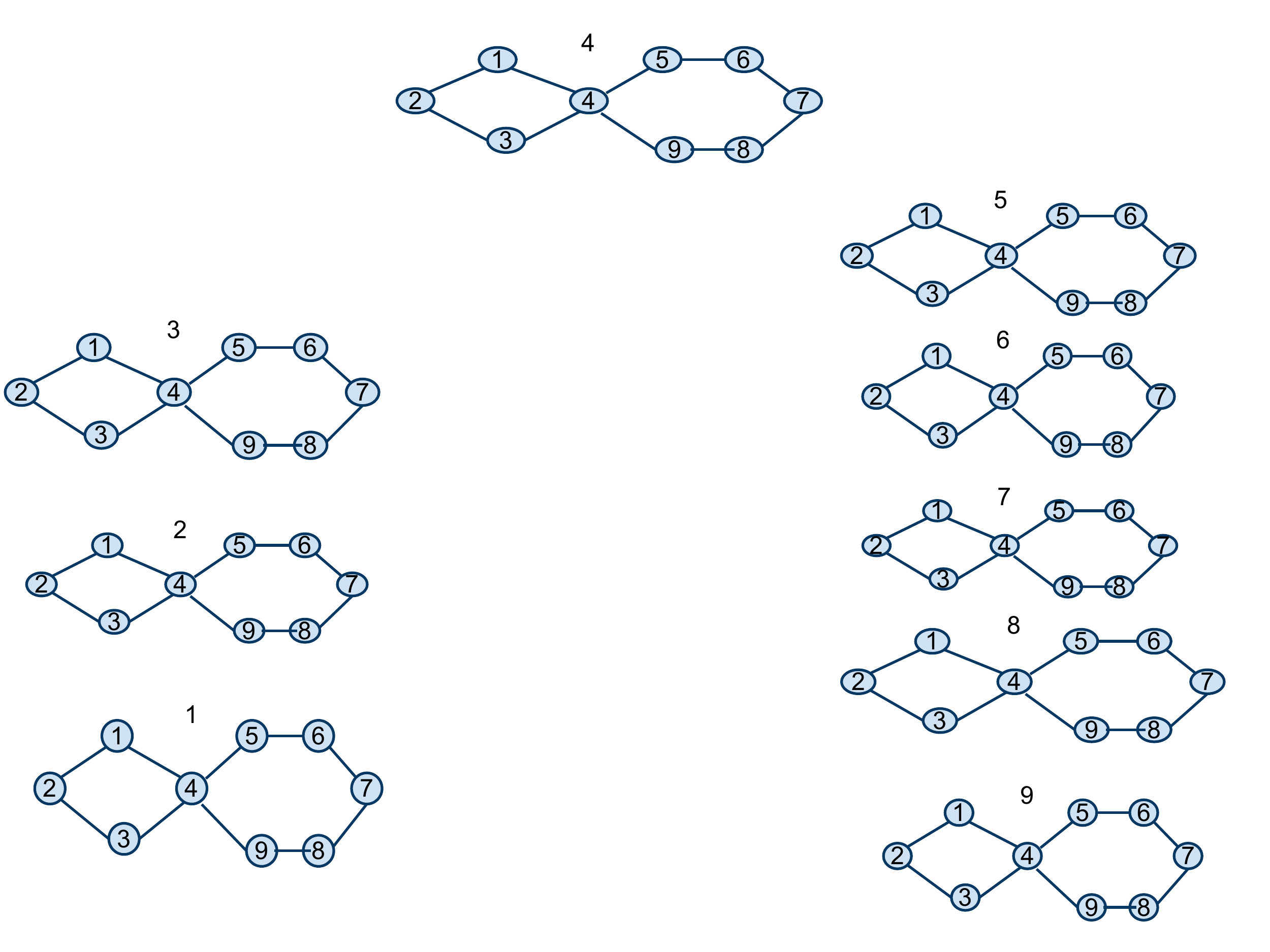}
\caption{$OTIS(BF(4,6))$: The intercluster edges are not shown to maintain clarity}
\end{center}
\end{figure}

We prove our claim in parts. Firstly we identify the forced edges in a Hamiltonian Cycle, assuming that one exists. Then we make a choice of picking one edge as Hamiltonian from an option of two, without any loss of generality. Finally, we arrive at a contradiction.

\claim{Vertex $\langle4,4\rangle$ cannot have both $(\langle4,1\rangle,\langle4,4\rangle)$ and $(\langle4,3\rangle,\langle4,4\rangle)$ as Hamiltonian edges.}

\proof{ If possible let both $\langle4,1\rangle$ and $\langle4,3\rangle$ be Hamiltonian neighbours of $\langle4,4\rangle$.Now, vertex $\langle4,2\rangle$ has to have either of $(\langle4,1\rangle,\langle4,2\rangle)$ and $(\langle4,2\rangle,\langle4,3\rangle)$ as its 
Hamiltonian edges. Without loss of generality, let $(\langle4,2\rangle,\langle4,3\rangle)$ be the Hamiltonian edge.This forces the following edges,
\begin{description}

 \item [Cluster 1:] $(2,3)$ 
\item[Cluster 2:]$(1,4)$
\item[Intercluster edges:] $(\langle4,1\rangle,\langle1,4\rangle), (\langle4,2\rangle,\langle2,4\rangle), (\langle1,3\rangle,\langle3,1\rangle)$
\end{description}
This in turn, forces the following edges:
\begin{description}
 \item [Cluster 3:] $(1,4)$ 
\item[Intercluster edges:] $(\langle3,2\rangle,\langle2,3\rangle)$

\end{description}

This clearly forms a forced subcycle. So both $(\langle4,1\rangle,\langle4,4\rangle)$ and $(\langle4,3\rangle,\langle4,4\rangle)$ cannot be Hamiltonian edges of $\langle4,4\rangle$.}

\claim{Vertex $\langle4,4\rangle$ cannot have both $(\langle4,1\rangle,\langle4,4\rangle)$ and $(\langle4,3\rangle,\langle4,4\rangle)$ as non-Hamiltonian edges.}
\proof{If possible, let both $(\langle4,1\rangle,\langle4,4\rangle)$ and $(\langle4,3\rangle,\langle4,4\rangle)$ edges be dropped from the Hamiltonian Cycle(assuming one exists.This forces the following edges:
\begin{description}
 \item [Cluster 4:] $(1,2), (2,3)$
 \item[Cluster 1:]$(2,3)$
 \item[Cluster 3:]$(1,2)$
 \item[Intercluster edges:] $(\langle3,1\rangle,\langle1,3\rangle), \langle4,1\rangle,\langle1,4\rangle), (\langle3,4\rangle,\langle4,3\rangle) $
\end{description}

This clearly forms a forced subcycle. So both $(\langle4,1\rangle,\langle4,4\rangle)$ and $(\langle4,3\rangle,\langle4,4\rangle)$ cannot be non-Hamiltonian edges of $\langle4,4\rangle$.}

Hence the only possibility is either of $(\langle4,1\rangle,\langle4,4\rangle)$ and $(\langle4,3\rangle,\langle4,4\rangle)$ is Hamiltonian edge. Without loss of generality (due to the symmetric structure), let $(\langle4,3\rangle,\langle4,4\rangle)$ be the Hamiltonian edge.
 
\begin{figure}[ht]
\begin{center}
\includegraphics[scale=.28]{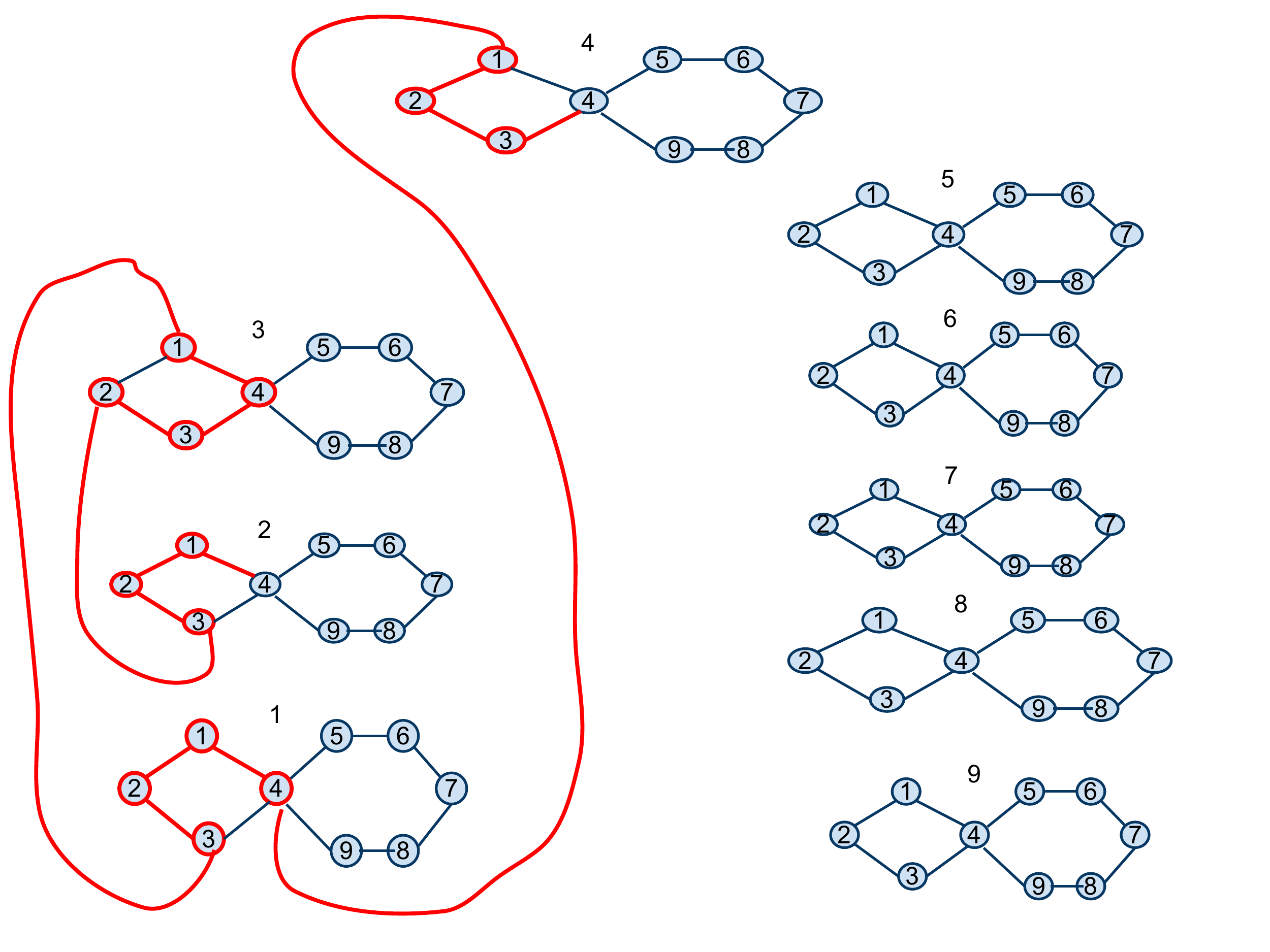}
\caption{$OTIS(BF(4,6))$:The vertices and edges shown in red are saturated vertices and forced edges respectively}
\end{center}
\end{figure}

\noindent{Now, in this figure, (Figure 8), the vertices which have already obtained both of its Hamiltonian neighbours, i.e., saturated vertices, are shown in red. We see that $\langle2,4\rangle$ and $\langle4,4\rangle$ are the vertices
which have obtained exactly one of its Hamiltonian neighbours and rest of the vertices are either completely saturated, or not saturated. So a Hamiltonian cycle exists if and only if, there exists a path between 
$\langle2,4\rangle$ and $\langle4,4\rangle$ spanning all the unsaturated vertices.}\\

\noindent{Vertex $\langle4,4\rangle$ has to have either of $\langle4,5\rangle, \langle4,9\rangle$ as its Hamiltonian neighbour. Without any loss of generality, let us assume $\langle4,9\rangle$ is its Hamiltonian neighbour. This again forces a set of edges and makes certain vertices saturated.   
In the following figure, the forced edges and saturated vertices are marked in red. (Figure 9)}

\begin{figure}[ht]
\begin{center}
\includegraphics[scale=.28]{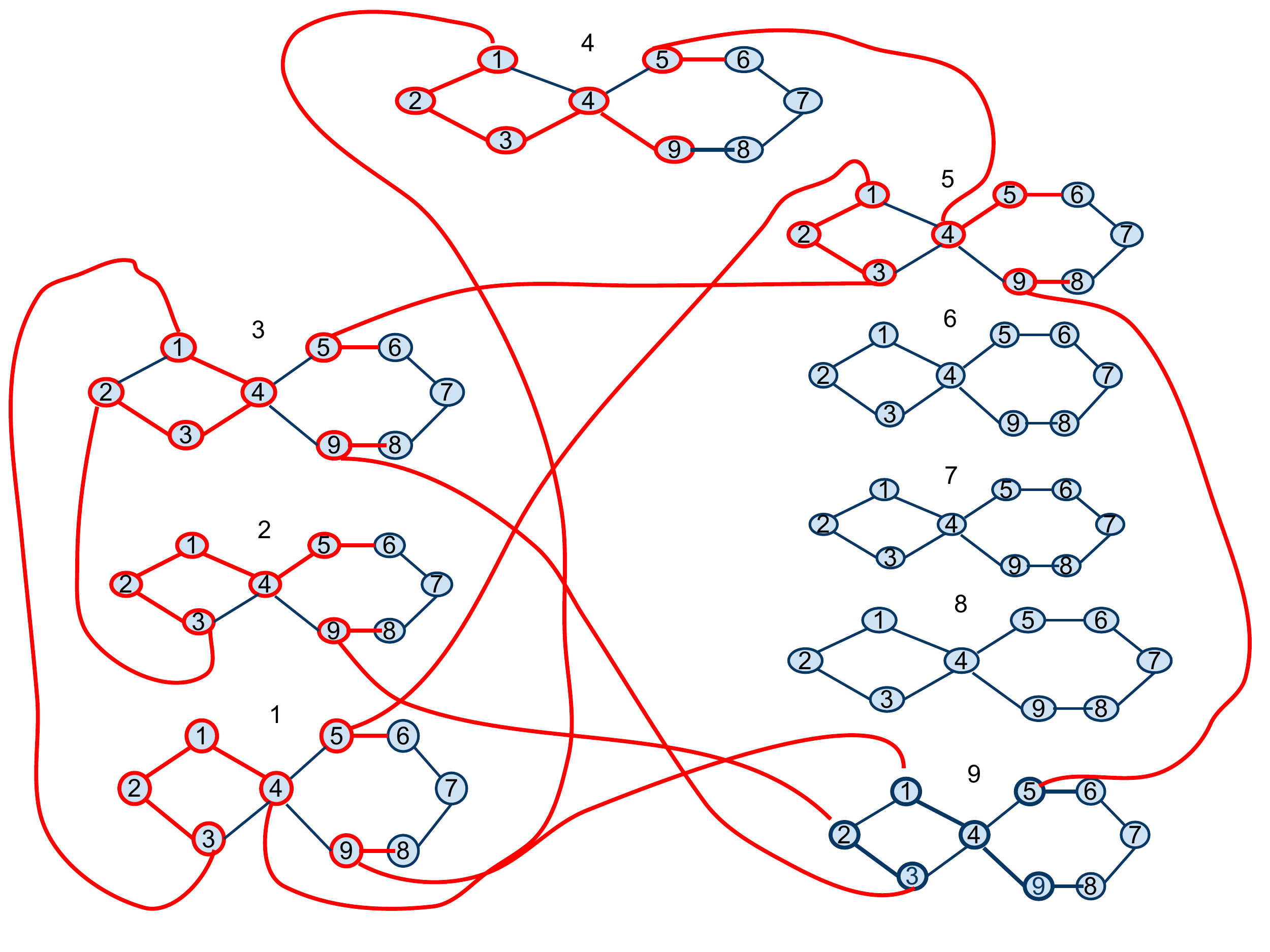}
\caption{$OTIS(BF(4,6))$:The vertices and edges shown in red are saturated vertices and forced edges respectively}
\end{center}
\end{figure}

\noindent{Now, we see in cluster 9, vertex 2, has to have either of $\langle9,1\rangle$ or $\langle9,3\rangle$ as it Hamiltonian edge. Without loss of generality, let $\langle9,3\rangle$ be its Hamiltonian edge. This again forces a set of edges.(Figure 10)}

\begin{figure}[ht]
\begin{center}
\includegraphics[scale=.28]{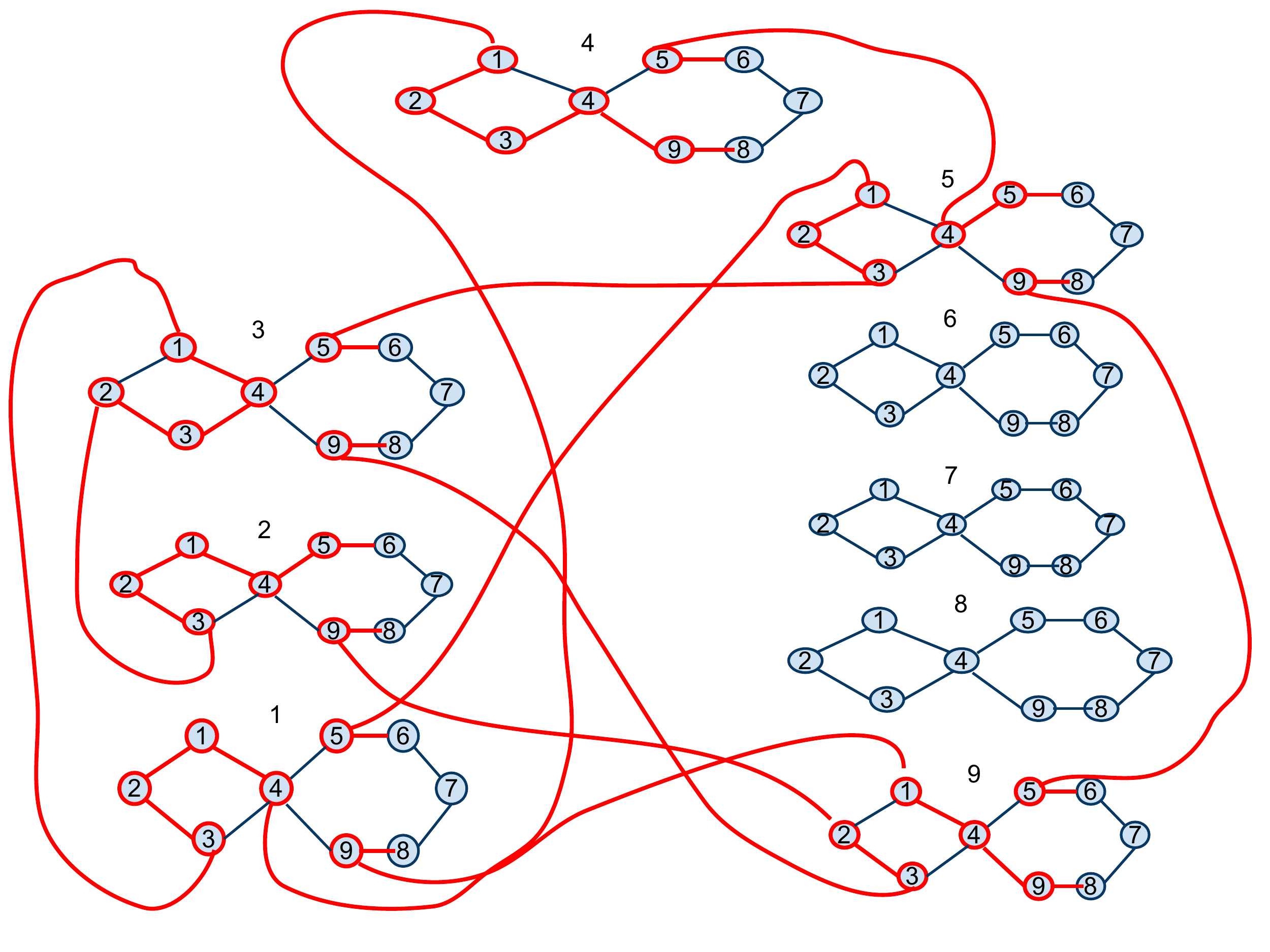}
\caption{($OTIS(BF(4,6))$:The vertices and edges shown in red are saturated vertices and forced edges respectively}
\end{center}
\end{figure}

\noindent{Now, vertex $\langle2,6\rangle$ has to adopt any one of $\langle2,7\rangle, \langle6,2\rangle$ as its Hamiltonian neighbours. Let us study both the cases.}
\begin{description}
 \item [\textbf{Case 1:}]
Let vertex $\langle2,6\rangle$ choose $\langle6,2\rangle$ as its Hamiltonian neighbour. Then the following edges become forced:
\begin{description}
 \item [Cluster 2:] $(7,8)$
 \item[Cluster 8:] $(2,1), (2,3)$
 \item[Intercluster edges:] $(\langle2,7\rangle,\langle$BF(2m+1,2n+1)$ or $BF(2m+1,2k)$7,2\rangle)$
\end{description}

Now, In cluster 6, vertex 2, has to have either of $\langle2,1\rangle$ and $\langle2,3\rangle$ as its Hamiltonian edge. Without loss of generality, let us take $\langle2,1\rangle$ as its Hamiltonian edge. This, in turn
forces the following edges:
\begin{description}
 \item [Cluster 6:]$(3,4)$.
 \item[Intercluster edges:] $(\langle3,6\rangle,\langle6,3\rangle)$ and hence $(\langle3,7\rangle,\langle7,3\rangle)$. This forces the following edges:
 \item[Cluster 3:] $(7,8)$
 \item[Cluster 8:]$(3,2), (3,4)$. Hence $(1,4)$ cannot be an edge in this cluster. This forces the following set of edges:
 \item[Intercluster edges:] $(\langle8,1\rangle, \langle1,8\rangle)$, and hence $(\langle1,7\rangle, \langle7,1\rangle)$ and the following edges:
 \item[Cluster 1:] $(7,6)$ 
 \item[Cluster 6:] $(1,4), (5,6), (8,9)$
 \item[Intercluster edges:] $(\langle5,6\rangle, \langle6,5\rangle)$ and $(\langle6,9\rangle, \langle9,6\rangle)$. This in turn forces the following edges:
 \item[Cluster 4:] $(6,7)$
 \item[Cluster 9:] $(7,8)$
 \item[Intercluster edges:] $(\langle7,9\rangle, \langle9,7\rangle)$
 \item[Cluster 8:] $(4,9)$. Note that, $(9,8)$ is already a forced edge as this edge is incident to a vertex of degree 2. Therefore, this forces the edge $(5,6)$, and the following edge:
 \item[Intercluster edge:] $(\langle5,8\rangle, \langle8,5\rangle)$

\end{description}

\noindent{Now, since both the vertices $\langle5,6\rangle$ and $\langle5,8\rangle$ have become saturated, the edges $(\langle5,6\rangle, \langle5,7\rangle)$ and $(\langle5,8\rangle, \langle5,7\rangle)$ have to be dropped from the Hamiltonian Cycle assuming one exists. 
This leaves vertex $\langle5,7\rangle$ with degree 1 and hence a Hamiltonian Cycle is not possible.} \\
  \item [\textbf{Case 2:}]{ Let vertex $\langle2,6\rangle$ choose $\langle2,7\rangle$ as its Hamiltonian neighbour. Then the edges $(\langle6,2\rangle, \langle6,1\rangle)$ and $(\langle6,2\rangle, \langle6,3\rangle)$ become forced. Now notice that in Cluster 6, both
the vertices $\langle6,1\rangle$ and $\langle6,3\rangle$ cannot have intercluster edges incident on then as Hamiltonian edges or as non-Hamiltonian edges, as both these cases forces subcycle formation. So exactly one of them
has to have the intercluster edge incident on it, as Hamiltonian edge. Without loss of generality, let vertex $\langle6,1\rangle$ has the edge $(\langle6,1\rangle, \langle1,6\rangle)$ as Hamiltonian. Now, this forces the following
edges:}

\begin{description}
 \item [Cluster 6:] $(3,4)$
 \item[Cluster 1:] $(7,8)$
 \item[Intercluster edges:] $(\langle1,7\rangle, \langle7,1\rangle)$
 \item[Cluster 8:] $(1,2), (1,4)$
 \item[Intercluster edges:] $(\langle3,7\rangle, \langle7,3\rangle)$ and $(\langle3,8\rangle, \langle8,3\rangle)$ [In Cluster 3, $(7,8)$ cannot be an edge, as this would force subcycle formation in Cluster 8]
\end{description}

\begin{figure}[ht]
\begin{center}
\includegraphics[scale=.26]{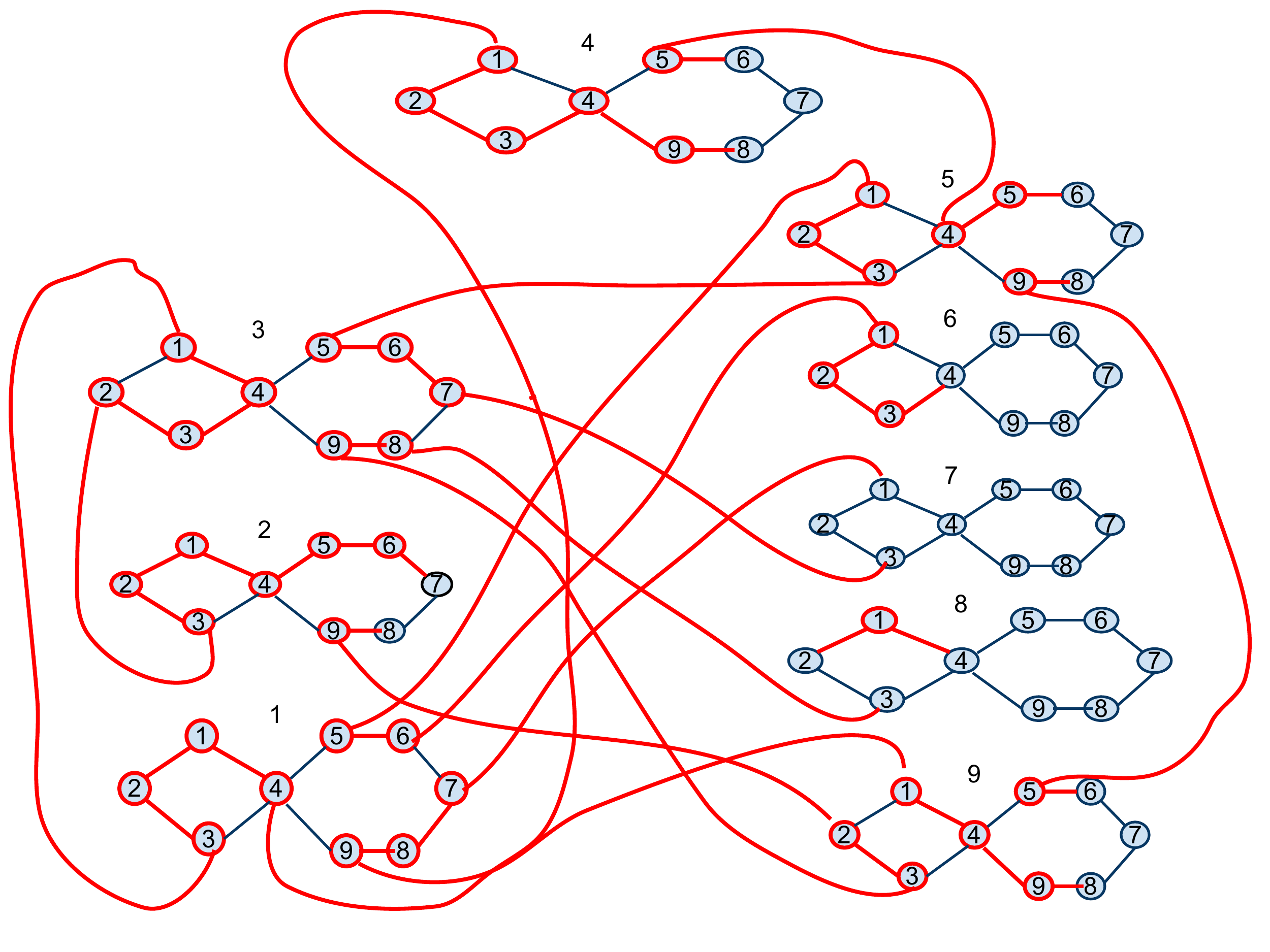}
\caption{$OTIS(BF(4,6))$:The vertices and edges shown in red are saturated vertices and forced edges respectively}
\end{center}
\end{figure}

\noindent{Now vertex $\langle5,7\rangle$ can choose any two of its three incident edges as Hamiltonian edges. e consider all three possible cases and show that Hamiltonian Cycle formation is impossible.}\\
\\
\noindent{\textbf{Case 1: $\langle5,7\rangle$ chooses $\langle5,6\rangle$ and $\langle5,8\rangle$ as its Hamiltonian neighbours.}}\\
{ In this case, the edges $(5,4)$ and $(5,6)$ becomes forced in Cluster 6, 7 and 8. In clusters 6 and 8, this saturates vertex 4. This, in turn forces the intercluster edges
$(\langle6,9\rangle, \langle9,6\rangle)$ and $(\langle8,9\rangle, \langle9,8\rangle)$. Now, since both the vertices $\langle9,6\rangle$ and $\langle9,8\rangle$ have become saturated, the edges $(\langle9,6\rangle, \langle9,7\rangle)$ and $(\langle9,8\rangle, \langle9,7\rangle)$ have to be dropped from the Hamiltonian Cycle assuming one exists. 
This leaves vertex $\langle9,7\rangle$ with degree 1 and hence a Hamiltonian Cycle is not possible.}\\

\noindent{\textbf{Case 2: $\langle5,7\rangle$ chooses $\langle5,6\rangle$ and $\langle7,5\rangle$ as its Hamiltonian neighbours.}}\\
{ In this case, the intercluster edge $(\langle5,8\rangle,\langle8,5\rangle)$ become forced. In cluster 6, edges $(4,5)$ and $(5,6)$  gets forced. This saturates vertex $\langle6,4\rangle$ hence forcing the following edges:
$(\langle4,6\rangle, \langle4,7\rangle)$, $(\langle6,9\rangle, \langle6,8\rangle)$ and $(\langle6,9\rangle, \langle9,6\rangle)$. This, in turn, forces edges $(\langle9,7\rangle, \langle9,8\rangle)$, $(\langle7,9\rangle, \langle9,7\rangle)$ and $(\langle8,4\rangle, \langle8,9\rangle)$.Now, vertex $\langle8,4\rangle$ gets saturated, hence forcing the
following edges in cluster 8: $(2,3)$ and $(5,6)$, and the edge $(7,8)$ in cluster 4. Now, note that $(\langle8,6\rangle, \langle6,8\rangle)$ cannot be a Hamiltonian edge, as it forces subcycle. Therefore, in cluster 8,
$(6,7)$ is an edge. This in turn forces the edge $(9,8)$ in cluster 7. Now in cluster 7, $(4,1)$ and $(4,3)$ becomes forced. 
Now, since both the vertices $\langle7,1\rangle$ and $\langle7,3\rangle$ have become saturated, the edges $(\langle7,2\rangle, \langle7,1\rangle)$ and $(\langle7,2\rangle, \langle7,3\rangle)$ have to be dropped from the Hamiltonian Cycle assuming one exists. 
This leaves vertex $\langle7,2\rangle$ with degree 1 and hence a Hamiltonian Cycle is not possible.}\\

\noindent{\textbf{Case 3: $\langle5,7\rangle$ chooses $\langle5,8\rangle$ and $\langle7,5\rangle$ as its Hamiltonian neighbours.}}\\
{Following similar argument, as above, we identify the following forced edges (Figure 12).}\\

\begin{figure}[ht]
\begin{center}
\includegraphics[scale=.26]{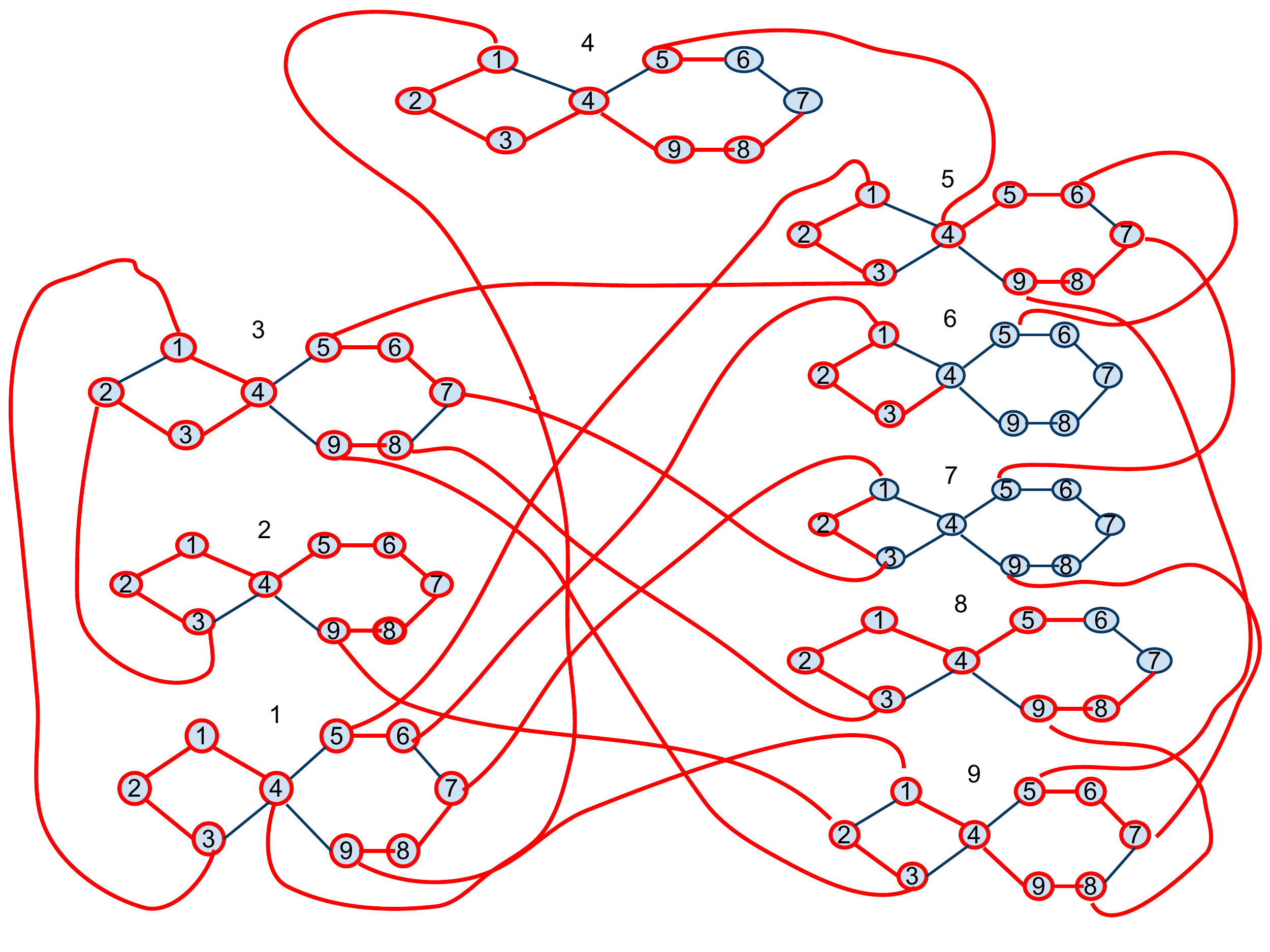}
\caption{$OTIS(BF(4,6))$:The vertices and edges shown in red are saturated vertices and forced edges respectively}
\end{center}
\end{figure}
\noindent{ Now, note that vertex $(\langle7,4\rangle$ cannot have both $\langle7,5\rangle$ and $\langle7,9\rangle$ as its Hamiltonian neighbours as this forces subcycle. Without loss of generality, let $\langle7,5\rangle$ be its Hamiltonian
neighbour, the other Hamiltonian neighbour being $\langle4,7\rangle$. This forces the following edges (figure 13).}  

\begin{figure}[ht]
\begin{center}
\includegraphics[scale=.26]{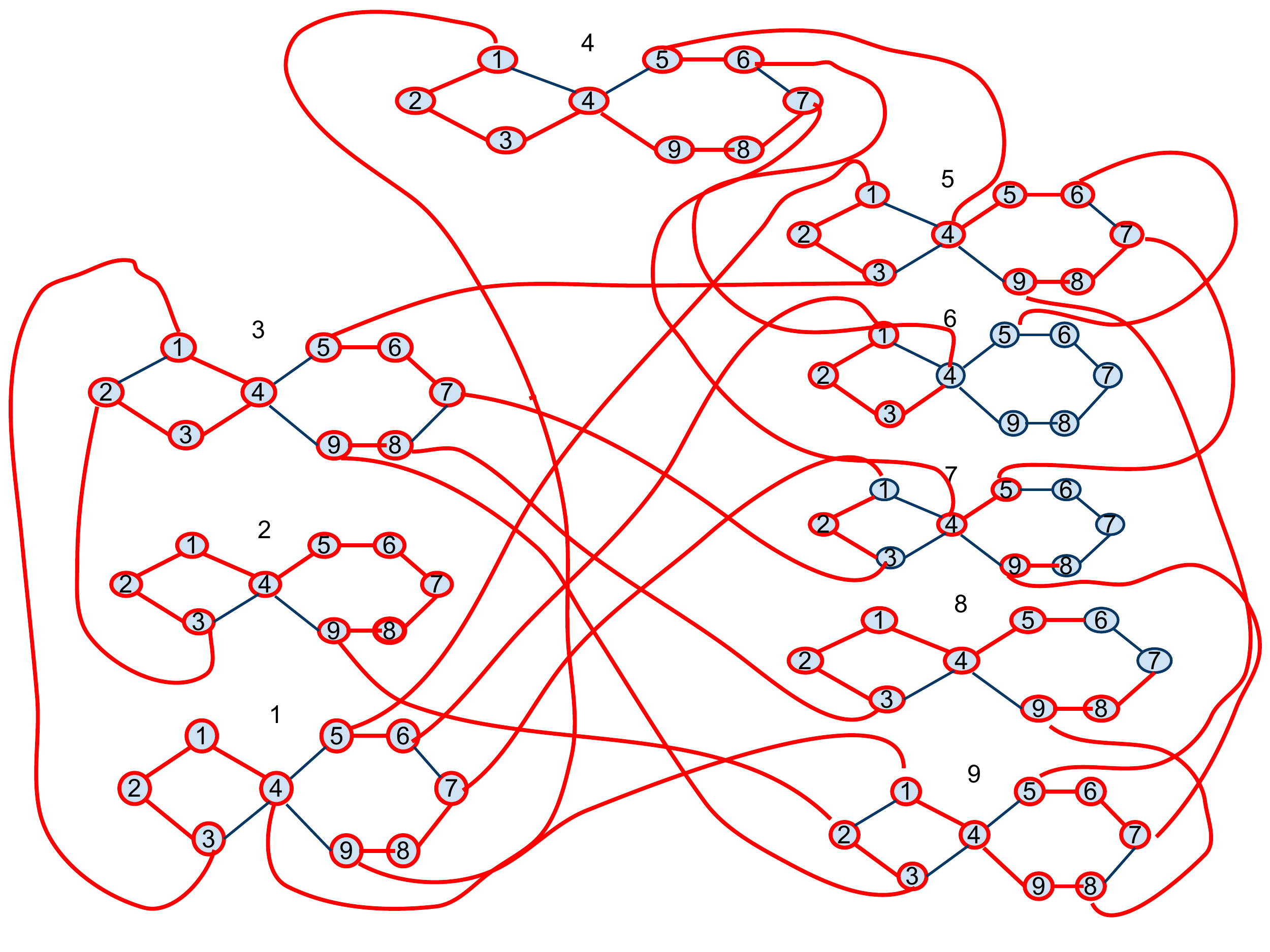}
\caption{$OTIS(BF(4,6))$:The vertices and edges shown in red are saturated vertices and forced edges respectively}
\end{center}
\end{figure}
\noindent{Now, since both the vertices $\langle6,4\rangle$ and $\langle9,6\rangle$ have become saturated, the edges $(\langle6,4\rangle, \langle6,9\rangle)$ and $(\langle6,9\rangle, \langle9,6\rangle)$ have to be dropped from the Hamiltonian Cycle assuming one exists. 
This leaves vertex $\langle6,9\rangle$ with degree 1 and hence a Hamiltonian Cycle is not possible.}

\end{description}
 This completes the proof.

\end {document}